\documentclass[11pt, a4paper]{article}

\usepackage[utf8]{inputenc}
\usepackage[margin=1in]{geometry}

\usepackage{amsmath}
\usepackage{amsfonts}
\usepackage{nicefrac}
\usepackage{rotating}
\usepackage{makecell}
\usepackage{multirow}
\usepackage[colorlinks=true, allcolors=blue]{hyperref}

\usepackage[oldcommands, ruled]{algorithm2e}
\SetKwFor{For}{for}{do}{}
\SetKwFor{While}{while}{do}{}

\usepackage{amsthm}
\theoremstyle{definition}
\newtheorem{defi}{Definition}[section]
\newtheorem{theorem}[defi]{Theorem}
\newtheorem{lemma}[defi]{Lemma}
\newtheorem{example}[defi]{Example}
\newtheorem{remark}[defi]{Remark}

\newcommand{\st}{\mathbin{|}}
\newcommand*{\calJ}{\mathcal J}

\usepackage{orcidlink}

\begin{document}

\author{Felix Buld\orcidlink{0009-0007-0230-8024}\thanks{Operations Research, Technical University of Munich, Arcisstr. 21, 80333 Munich, Germany} \and Andreas S. Schulz\orcidlink{0000-0002-9595-459X}$\footnotemark[1]$}

\title{Scheduling with Testing: Competitive Algorithms for Minimizing the Total Weighted Completion Time in the Adversarial Model}
\maketitle              
\begin{abstract}
We study scheduling with testing on a single machine and on identical parallel machines to minimize the total \emph{weighted} completion time in the adversarial model. In this setting, each job is equipped with a weight, an upper bound on its processing time, and a testing time. An algorithm can either execute a job for an amount of time equal to the upper bound or test it first to reveal a potentially lower processing time used to schedule the job later. 

We establish the first constant-competitive algorithms for this problem with job-depen\-dent weights that reflect each job's relative importance. For single-machine scheduling, we present a deterministic algorithm with a competitive ratio of 2.3166 and show that a randomized variant has a competitive ratio of 2.1523. These guarantees match the best-known upper bounds in the unweighted setting. Combining these algorithms with list scheduling yields competitive ratios of 2.7763 and 2.5110 for identical-parallel-machine scheduling, improving the previously best-known bounds even in the unweighted case.\footnote{A conference version of this paper \cite{buld2025scheduling} appeared in the proceedings of IJTCS-FAW 2025. This is the full version, including complete proofs and additional details throughout the paper.}
\end{abstract}
    
\quad\small{\textbf{Keywords:} Testing $\cdot$ Explorable Uncertainty $\cdot$ Online Algorithms $\cdot$  
\\\hspace*{8.25em}
Scheduling $\cdot$ Competitive Analysis $\cdot$ Weighted Total Completion Time}

\section{Introduction}
An emerging research branch in scheduling and beyond considers \emph{explorable uncertainty}, where uncertainty in data can be reduced through actions such as conducting tests. However, these actions typically come at a cost, such as money or time. The problem considered here, \emph{scheduling with testing}, belongs to this category. We are given jobs that must be assigned to machines, servers, or personnel over time. The processing times of jobs are initially unknown but can be determined through additional time-consuming individual tests. Therefore, one has to balance the benefit of additional certainty against the cost incurred from exploration. This framework is motivated by various real-world applications, see, e.g.,  \cite{albers2021explorable,durr2020adversarial,levi2019scheduling} for comprehensive overviews. One prominent example is \emph{medical care}, where patients are tested to determine their needed treatment. For instance, \cite{sun2018patient} considers emergency situations following mass casualty events, where limited medical personnel handle both the testing and the treatment. Another example is \emph{fault diagnosis} in \emph{maintenance}, where tests are conducted to explore the duration of repairs. In a further application, \emph{code optimizers} can be used in advance to potentially reduce the runtime of computer programs.

\subsection{Problem Description}
\label{subsec: Problem Description}
We examine the following adversarial scheduling with testing problem that is a mixture of both online and robust optimization frameworks and was introduced in \cite{durr2020adversarial}. We are given $n$ jobs to be scheduled on a single machine. Each job $j$ has a weight $w_j > 0$ expressing its relative importance, which is known in advance. There are two possible processing options per job. The first is to execute the job untested in a safe mode, requiring $u_j>0$ time. The second option is to test a job, taking~$t_j> 0$ time, which reveals the actual processing time $p_j\in \left[ 0, u_j\right]$, which is initially unknown. The job can then be executed taking~$p_j$ time at any later point. The completion time of a job~$j$ in a schedule is denoted as~$C_j$. The objective is to minimize the total weighted completion time, $\sum_j w_j C_j$. It is important to note that at any point in time the machine can test or process at most one job at a time, i.e., testing requires the same capacity as processing.

The classic counterpart--without testing--is the scheduling problem known as $1 \st \st  \sum w_j C_j$, which considers~$n$ jobs with known weights $w_j$ and processing times~$p_j$. This problem is optimally solved by  \emph{Smith’s rule}, aka the \emph{Weighted Shortest Processing Time First} (\textsc{WSPT}) rule~\cite{Smith1956WSPT}. That is, in an optimal sequence jobs appear in order of non-decreasing ratios~$\nicefrac{p_j}{w_j}$. This result holds even if job preemptions were allowed, since these can be iteratively eliminated without increasing the objective value of the schedule under consideration, see, e.g., \cite{brucker2007scheduling}. Furthermore, the objective value decreases monotonically as each job's processing time decreases. If the processing times~$p_j$ were fully known in the setting with testing, an optimal solution would therefore be obtained by simply testing all jobs with~$t_j + p_j \leq u_j$ and processing the jobs in non-decreasing order of~$\nicefrac{1}{w_j} \cdot \min\lbrace t_j + p_j, u_j\rbrace$, ensuring that the execution of a tested job directly follows its test. However, in the model considered here, a job’s processing time~$p_j$ is only revealed after testing. The goal is to devise a (randomized) $\gamma$-competitive algorithm whose outcome--for any problem instance--has an (expected) objective value at most $\gamma$ times the objective value of an optimal solution that would be achievable with full information. For an instance~$I$ of the problem setting at hand, let $ALG(I)$ and $OPT(I)$ denote the obtained objective values of an algorithm under consideration, and an optimal offline solution, respectively. If the instance is clear from the context, we may abbreviate these by $ALG$ or $OPT$.

\subsection{Related Work}
\label{subsec: related work}
The problem with $w_j=1$ for all jobs $j$ has been studied under three different settings: preemptive, test-preemptive, and non-preemptive, as discussed in~\cite{albers2021explorable}. Accordingly, the (testing and) processing of a job may be interrupted at any point in time and resumed later, or only between testing and processing, or not at all. The most commonly considered setting is the test-preemptive one that we will always consider, unless stated otherwise. This work introduces the first competitive algorithms for this scheduling problem with arbitrarily given weights~$w_j > 0$. 
{\small
\begin{table}[htp]
    \centering
    \caption{Research on scheduling with testing to minimize the total (weighted) completion time in the adversarial model.}
    \label{tab:testing_literature}
    \vspace{2ex}
    \begin{tabular}{c|c|c|c|c}
         Model&  Objective & References  & Testing Time & Hidden Information \\
         \hline
         \multirow{4}{*}{Adversarial} & \multirow{2}{*}{$\sum_j C_j$}  &\cite{durr2020adversarial}   &Uniform & \multirow{2}{*}{$p_j\in \left[ 0, u_j\right]$}\\
         &&\cite{albers2021explorable,gong2024approximation,liu2023power,liu2024parallel} &Job-dependent & \\
         \cline{2-5}
         & \multirow{2}{*}{$\sum_j w_j C_j$} &\cite{zhang2023scheduling}  &Uniform&$p_j\in \left[ 0, u\right]$\\
         && This work   &Job-dependent& $p_j\in \left[ 0, u_j\right]$ 
    \end{tabular}
\end{table}
}

Scheduling with testing has so far also been explored in various settings, including stochastic optimization variants \cite{levi2019scheduling,levi2024scheduling,sun2018patient}, speed scaling \cite{bampis2021speed}, makespan minimization on parallel machines \cite{albers2021multiplemachines,gong2026randomized, gong2025multiprocessor,gong2022improved}, processing time oracles \cite{dufosse2022scheduling}, testing budgets \cite{damerius2023scheduling}, obligatory testing \cite{dogeas2024obligatory,liang2026completion}, and flow time minimization~\cite{lindermayr2026online}.

Table~\ref{tab:testing_literature} provides an overview of research on the model from Subsection~\ref{subsec: Problem Description}. Early research on the single-machine variant focused on uniform testing times \cite{durr2020adversarial}, but this work builds on the more general framework including general job-dependent testing times, as in \cite{albers2021explorable,gong2024approximation,liu2023power,liu2024parallel}.
Table~\ref{tab:bounds on a single machine} summarizes the achieved bounds for this model on a single machine, where the \emph{Golden Ratio} is denoted by $\varphi := \nicefrac{\left( 1 + \sqrt{5}\right)}{2}\approx 1.618$. Objectives studied include the makespan $C_{\max}:=\max_j C_j$ and the total (weighted) completion time.
{\small
\begin{table}[htb]
    \centering
    \caption{Bounds for adversarial single-machine scheduling with testing.}
    \label{tab:bounds on a single machine}
    \vspace{2ex}
    \begin{tabular}{c|c|c|c|c}
         Objective & Type of Algorithm & \makecell{Lower Bound \\ for $t_j=1$} & \makecell{Upper Bound \\ for $t_j=1$} & \makecell{Upper Bound \\ for $t_j$ arbitrary}  \\
         \hline
         \multirow{2}{*}{$C_{\max}$} & Deterministic  & $\varphi$  \cite{durr2020adversarial} & $\varphi$ \cite{durr2020adversarial} & $\varphi$ \cite{albers2021explorable}\\
         & Randomized & $\nicefrac{4}{3}$ \cite{durr2020adversarial} & $\nicefrac{4}{3}$ \cite{durr2020adversarial} & $\nicefrac{4}{3}$ \cite{albers2021explorable}\\
         \hline
         \multirow{2}{*}{$\sum_j C_j$} & Deterministic & $1.8546$ \cite{durr2020adversarial} & $2$ \cite{durr2020adversarial}& $2.3166$ \cite{liu2023power}\\
          & Randomized & $1.6257$ \cite{durr2020adversarial} & $1.7453$ \cite{durr2020adversarial}& $2.1523$ \cite{liu2023power}
    \end{tabular}
\end{table}
}
\\Minimizing the total completion time on identical parallel machines has been considered in~\cite{gong2024approximation,liu2024parallel}. The currently best-known upper bounds for the unweighted case are $3.2361$ for a deterministic algorithm and $2.8307$ for a randomized algorithm, established in~\cite{gong2024approximation}. Recently, \cite{liu2024parallel} presented ideas for an improved upper bound of $2.7763$ for a deterministic algorithm, which this paper develops into a complete proof and extends to randomized algorithms. The work of~\cite{zhang2023scheduling} made the first attempt to address minimizing the total \emph{weighted} completion time on a single machine in the adversarial model. However, the algorithms outlined by \cite{zhang2023scheduling} do not admit bounded competitive ratios; we provide adversarial instances proving this in Appendix~\ref{subsec: unboundedness zhang}. We show how to appropriately integrate weights to achieve constant-competitive algorithms. Our main contributions in this paper are:
\begin{itemize}
    \item[] A deterministic $2.3166$-competitive and a randomized $2.1523$-competitive algorithm for minimizing $\sum_j w_jC_j$ on a single machine.
    \item[] A deterministic $2.7763$-competitive and a randomized $2.5110$-competitive algorithm for minimizing $\sum_j w_jC_j$ on identical parallel machines.
\end{itemize}

\subsection{Paper Outline}
This paper advances research on scheduling with testing in the adversarial model and under the total \emph{weighted} completion time objective. In Section~\ref{section: From Preemptive to Test Preemptive Scheduling}, we demonstrate how the \textsc{Golden Round Robin Rule} from~\cite{albers2021explorable} can be naturally extended to account for job-dependent weights. We show that its test-preemptive version coincides with the weighted version of the $\left(\alpha, \beta\right)$-\textsc{PCP} algorithm of~\cite{liu2023power} for a certain choice of parameters and prove an upper bound for the latter leveraging this connection. In Section~\ref{sec: amortized analysis and delaying testing tasks}, we extend the best-known deterministic and randomized versions of \cite{liu2023power} to incorporate job-dependent weights. With this approach, the currently best-known competitive ratios are preserved in the weighted setting. In Section~\ref{section: Parallel Machines}, we extend the single-machine algorithms to identical parallel machines. The resulting deterministic performance guarantee matches that of the best-known deterministic algorithm for uniform weights. The competitive ratio of the randomized algorithm even improves upon that of the best randomized algorithm known for the case of uniform weights.

\section{From Preemptive to Test-Preemptive Scheduling}
\label{section: From Preemptive to Test Preemptive Scheduling}
In single-machine scheduling with testing, two kinds of decisions are required: which jobs to test and in what order to schedule the tasks. The standard approach in the adversarial model has been to decouple those decisions. The decision of which jobs to test is solely based on thresholds, and the ordering decisions depend on test outcomes. We follow this regime as well. In Subsection~\ref{subsec: testing a single job}, we (re)state the main results for testing and makespan minimization and discuss them considering job-dependent weights. In Subsection~\ref{subsec: pcp algorithm} we present a \emph{weighted} version of the~$\left(\alpha, \beta\right)$-\textsc{PCP} algorithm from~\cite{liu2023power}. In Subsection~\ref{subsec: wrr preemptive}, we introduce a \emph{weighted} version of the \textsc{Golden Round Robin Rule} from \cite{albers2021explorable} and link it to the previous algorithm by making it test-preemptive, allowing us to transfer upper bounds.

\subsection{Preliminaries: Testing a Single Job and Weighted Makespan Minimization}
\label{subsec: testing a single job}
As a preliminary step, we revisit makespan minimization in the adversarial scheduling with testing model. The motivation is that the resulting analysis will later be used to guide testing decisions and bound algorithmic processing times throughout the paper. As a by-product, we extend the best-possible deterministic algorithm for minimizing the makespan to the setting of \emph{weighted} makespan minimization. This objective, i.e., minimizing $\max_j w_j C_j$, has recently also been studied in related online scheduling models, e.g., in \cite{amouzandeh2025minimizing}. 

The best-possible deterministic and randomized algorithms for minimizing the makespan on a single machine boil down to making optimal testing decisions for single jobs, since each job contributes linearly to the (expected) makespan, see \cite{albers2021explorable,durr2020adversarial}. In particular, it holds that $\max_{j} C_j^A = \sum_{j=1}^n p_j^A$.

The optimal processing time of job~$j$ is given by $p_j^*:= \min\lbrace t_j + p_j, u_j \rbrace$ and the algorithmic processing time is denoted by $p_j^A$, which is either $t_j+p_j$ or $u_j$, depending on whether the algorithm decides to test. The basic deterministic rule is to test a job if and only if its upper-bound-to-testing-time ratio $r_j:=\nicefrac{u_j}{t_j}$ is at least a chosen threshold $\alpha$ fixed for all jobs in advance, and then scheduling jobs feasibly and without interruption. 

\begin{lemma}[\cite{albers2021explorable,durr2020adversarial}]
    \label{lemma: single job testing}
    If job $j$ is tested, then $p_j^A\leq \left(1+ \nicefrac{1}{r_j}\right)p^*_j \leq\left(1+ \nicefrac{1}{\alpha}\right)p^*_j$. If~$j$ is not tested, then $p_j^A \leq r_j p^*_j \leq \alpha p^*_j$. Testing job~$j$ if and only if $r_j \geq \varphi$ leads to a~$\varphi$- competitive algorithm for minimizing the makespan. This is best possible for deterministic algorithms.
\end{lemma}
One especially considers the two extreme outcomes $p_j=0$ if an algorithm does not test, or $p_j=u_j$ if an algorithm decides to test job~$j$. Randomizing the testing decision balances between these two worst cases and yields improved guarantees with matching lower bounds derived with Yao's principle:
\begin{lemma}[\cite{albers2021explorable,durr2020adversarial}]
    \label{lemma: randomized single job testing}
    The best-possible randomized algorithm for minimizing the makespan tests a job~$j$ with probability $\nicefrac{\left(r_j^2-r_j\right)}{\left(r_j^2 -r_j + 1\right)}$, if $r_j>1$, and does not test otherwise. It achieves $E[p_j^A]\leq\nicefrac{4}{3}\cdot p_j^*$ and a competitive ratio of~$\nicefrac{4}{3}$.
\end{lemma}
We now discuss the result on deterministic algorithms from the viewpoint of weighted makespan minimization.
\begin{theorem}
\label{theorem: weighted makespan deterministic}
    Scheduling jobs non-preemptively in order of non-increasing weights combined with testing decisions based on the threshold as in Lemma~\ref{lemma: single job testing}, yields a $\varphi$-competitive deterministic algorithm for minimizing $\max w_j C_j$.
\end{theorem}
\begin{proof}
    Let jobs be ordered such that $w_1\geq w_2\geq\ldots\geq w_n$. For known processing times, preemption is not beneficial, and an optimal offline algorithm processes job~$j$ for $p_j^*$ units. This can be shown job-by-job, in decreasing order of completion time, by shifting the parts of a job directly ahead of its last part, and all other jobs occupying this space, to the front accordingly. This procedure does not increase any completion time or the objective value and ensures a non-preemptive schedule. Reducing the processing times to $p_j^*$ for each job $j$ does not increase the objective value either.
    
    We next argue, by an exchange argument (which can be applied iteratively), that it is optimal for an algorithm to schedule jobs in order of non-increasing weights, independent of processing times. Let~$l$ and~$k$ be two consecutively scheduled jobs with weights $w_l< w_k$ and completion times $C_l<C_k$. By swapping these two jobs, we get a new schedule with completion times $C_k'$ and $C_l'$, with $\max\lbrace w_k C_k',w_l C_l'\rbrace\leq w_k C_k$ and thus no worse objective value.
    Therefore, sorting by weights and using the testing decision as in Lemma~\ref{lemma: single job testing}, we obtain for the deterministic setting:
    \begin{equation*}
        ALG = \max_{j=1,\ldots,n} w_j \sum_{k\leq j}p_k^A \leq  \max_{j=1,\ldots,n} w_j  \sum_{k\leq j}\varphi p_k^* = \varphi \cdot OPT \ .
    \end{equation*}
\end{proof}

The algorithms developed in the following sections rely on testing decisions based on thresholds or probability functions, as described in this section. In contrast to weighted makespan minimization, where an optimal job ordering can be fixed in advance, the main challenge for minimizing $\sum w_j C_j$ in this model, apart from testing decisions, concerns dynamic ordering decisions made after processing times are gradually revealed. Unlike the strategies in \cite{zhang2023scheduling} for this objective, which either order jobs by non-decreasing weights too greedily or do not incorporate the weights into comparisons between testing and execution tasks, see Appendix~\ref{subsec: unboundedness zhang}, we provide scheduling rules that integrate weights appropriately to yield competitive algorithms.

\subsection{An Algorithm for Test-Preemptive Scheduling}
\label{subsec: pcp algorithm}
In this subsection, we introduce the \textsc{Weighted} $\left(\alpha, \beta\right)$-\textsc{PCP} algorithm, which is a generalization of the $\left(\alpha, \beta\right)$-\textsc{PCP} algorithm of \cite{liu2023power} for the unweighted case, see Algorithm~\ref{algo: weighted pcp}. The abbreviation \textsc{PCP} was introduced in \cite{liu2023power} and stands for \emph{Prioritizing-Certain-Processing} (times) (not to be confused with Probabilistically Checkable Proofs). First, the testing decision is made independently for each job according to a threshold~$\alpha$, as in Section~\ref{subsec: testing a single job}. Secondly, we keep a priority list, containing currently available tasks. These are sorted according to priority values, which are~$\nicefrac{\beta t_j}{w_j}$ for the testing task (and some delay factor $\beta> 0$), $\nicefrac{u_j}{w_j}$ if a job is executed without testing, and~$\nicefrac{\left(t_j+p_j\right)}{w_j}$ for the execution of a job~$j$ that was tested. In each iteration, the task with the smallest proportionally weighted priority value is scheduled next, ties broken arbitrarily. After completing a testing task for a job, the processing time is revealed, and a corresponding tested-execution task is added.

\begin{algorithm}[htb]
\caption{\textsc{Weighted} $\left(\alpha, \beta\right)$-\textsc{PCP}}
    \label{algo: weighted pcp}
    \KwIn{Jobs $\calJ$, a threshold $\alpha\geq 1$, a delay factor $\beta > 0$.}
    \KwOut{A feasible test-preemptive schedule on a single machine.}
    \BlankLine
    Jobs to be tested: $\calJ_T := \lbrace j \in \calJ \st \alpha t_j \leq  u_j \rbrace$; Jobs left untested $\calJ_U := \calJ\setminus\calJ_T$;
    \\$L :=$ Empty priority list with priority values $\sigma_j$;  
    \\Add tasks for jobs from $\calJ_T$ with $\sigma_j =\nicefrac{\beta t_j}{w_j}$, from $\calJ_U$ with $\sigma_j =\nicefrac{u_j}{w_j}$ to $L$;
    \BlankLine
    \While{$L\neq\emptyset$} 
    {
        Let $j_{\min}$ be the job index of a task in $L$ with smallest $\sigma_j$; \\ 
        Process $j_{\min}$ and remove it from $L$;
        \\\uIf{the task was a testing task}
        {Reinsert a task into $L$ with $\sigma_{j_{\min}} = \nicefrac{\left(t_{j_{\min}} + p_{j_{\min}}\right)}{w_{j_{\min}}}$.}
    }
\end{algorithm}
The intuition behind the priority values in the algorithm for $\beta=1$ is that the jobs are completed in \textsc{WSPT}-order w.r.t.\ the algorithmic processing times $p^A_j$, as we will see in Subsection~\ref{subsec: wrr preemptive}. Note that if all testing decisions were correct, i.e., $p_j^A=p_j^*$ for all jobs $j$, jobs are completed in an optimal order. However, the jobs are not necessarily started in this order, since full information about a job is revealed only upon completion of its testing task. Another observation is that all appearing tasks are completed in order of their priority values $\sigma_j$ for the choice of $\beta=1$.

In Subsection~\ref{subsec: wrr preemptive}, we focus on $\beta=1$ and prove this section's main result, Theorem~\ref{theorem: pcp algo 2varphi}, based on the analysis of the \textsc{Weighted Golden Round Robin Rule} for the related preemptive setting.
\begin{theorem}
    \label{theorem: pcp algo 2varphi}
    The \textsc{Weighted} $\left(\varphi, 1\right)$-\textsc{PCP} algorithm is $2\varphi$-competitive.
\end{theorem}

In Section~\ref{sec: amortized analysis and delaying testing tasks}, we will explain the intuition behind the delay factor $\beta$ as introduced by \cite{albers2021explorable} and show an improved upper bound for another value of $\beta$, leveraging a refined analysis technique. 

\subsection{\textsc{Weighted Golden Round Robin} for Preemptive Scheduling}
\label{subsec: wrr preemptive}
Consider the \textsc{Golden Round Robin Rule} from \cite{albers2021explorable}. It tests job $j$ if its ratio~$r_j$ is at least $\alpha = \varphi$ and performs the usual \textsc{Round Robin Rule} \cite{motwani1994nonclairvoyant}, but on the algorithmic processing times, cycling through all unfinished tasks and processing each for an equal amount of time until all are completed. Using the  \textsc{Weighted Round Robin Rule} ($\text{WR}^3$) \cite{kim2003non}  on the algorithmic processing times as a subroutine, in which all jobs are processed simultaneously in proportion to their weights, leads to the \textsc{Weighted Golden Round Robin Rule} as described in Algorithm~\ref{algo:golden wrr}.

\begin{algorithm}[htb]
    \caption{\textsc{Weighted Golden Round Robin Rule} ($\text{WGR}^3$)}
    \label{algo:golden wrr}
    \KwIn{Jobs $\calJ$.}
    \KwOut{A feasible preemptive schedule on a single machine.}
    \BlankLine
    Jobs to be tested: $\calJ_T := \lbrace j \in \calJ \st \varphi t_j \leq  u_j \rbrace$; Jobs left untested $\calJ_U := \calJ\setminus\calJ_T$;
    \\$L :=$ Set of tasks (weight $w_j$ and time $t_j$ for job $j\in \calJ_T$ and $w_j$, $u_j$ for~$j\in \calJ_U)$.
    \BlankLine
    \While{$L\neq\emptyset$} 
    {
       Process the task of job $k$ at rate $\nicefrac{w_k}{\sum_{j\in L} w_j}$ for each $k\in L$ until a task finishes; Let $j_{\min}$ be its job index; Remove it from $L$;
       \\\uIf{the task of job~$j_{\min}$ was a testing task}
        {Reinsert a task into $L$ with weight $w_{j_{\min}}$ and processing time $p_{j_{\min}}$.
        }}
\end{algorithm}
It was already pointed out in \cite{dogeas2024obligatory,durr2020adversarial,gong2024approximation,liu2023power} that algorithms generating preemptive schedules can be made test-preemptive without increasing the objective value. We explicitly make Algorithm~\ref{algo:golden wrr} test-preemptive to connect the research of \cite{albers2021explorable} and \cite{liu2023power} and to show Theorem~\ref{theorem: pcp algo 2varphi}. For the analysis of Algorithms~\ref{algo: weighted pcp} and \ref{algo:golden wrr}, we define the following relations. For known $p_j$, it is optimal to process the jobs in order of non-decreasing $\nicefrac{p^*_j}{w_j}$. Consider a fixed optimal order. We write $k <_o j$ if job~$k$ appears in this order before job $j$. In addition, we fix the order of jobs in which Algorithm~\ref{algo: weighted pcp} (with $\alpha=\varphi$ and $\beta=1$) completes them. For this order, we write $k<_A j$, if job~$k$ finishes before job~$j$. Note that the jobs are completed in order of non-decreasing ratios~$\nicefrac{p_j^A}{w_j}$.

\begin{lemma}
\label{lemma: non-decreasing completion times}
    Let $C_j^{\text{WGR}^3}$ be the completion time of job $j$ in Algorithm~\ref{algo:golden wrr}. 
    Consider two jobs $j\not = k$. We have that $C_j^{\text{WGR}^3}\leq C_k^{\text{WGR}^3}$ if and only if $ \nicefrac{p_j^A}{w_j} \leq \nicefrac{p_k^A}{w_k}$. If~$j$ is completed and~$k$ is not completed yet, the remaining processing time of $k$ at time~$C_j^{\text{WGR}^3}$ is $p_k^A-\nicefrac{w_k}{w_j}\cdot p_j^A \geq 0$. Especially, we obtain $C_j^{\text{WGR}^3}=\sum_{k\leq_{A} j} p_k^A + \sum_{k>_{A} j} \nicefrac{w_k}{w_j} \cdot p_j^A$. 
\end{lemma}
\begin{proof}
    The proof follows the standard arguments in \cite{kim2003non,pinedo2022scheduling}, observing that when job~$j$ completes, job~$k$ has received $\nicefrac{w_k}{w_j} \cdot p_j^A$ units of processing.
\end{proof}

(Partly) processing a job $k$ before another job $j$ is finished contributes to the completion time of the latter and, therefore, to the total cost $\sum_j w_j C_j$. We define $c(k,j)$ as the \emph{contribution} of job $k$ to the completion time of job $j$, i.e., the amount of time spent on~$k$ before~$j$ is finished. The total cost therefore can also be expressed as $\sum_j \sum_k w_j \cdot c(k,j)$. We denote the optimal value (i.e., full information) by $OPT$.

\begin{lemma}
    \label{lemma: Golden WRR}
    Algorithm~\ref{algo:golden wrr} is $2\varphi$-competitive. This bound is tight.
\end{lemma}
\begin{proof}
    Consider the (preemptive) schedule constructed by Algorithm~\ref{algo:golden wrr}. 
    As in~\cite{albers2021explorable}, but generalized to job-dependent weights, we obtain for a job~$j$ using Lemma~\ref{lemma: non-decreasing completion times}:
    \begin{equation*}
         \sum_k w_j \cdot c(k,j) = \sum_{k\leq_A j}  w_j \cdot p_k^A + \sum_{k>_A j}  w_j \cdot \frac{w_k}{w_j}\cdot p_j^A = \sum_{k\leq_A j}  w_j \cdot p_k^A + \sum_{k>_A j} w_k \cdot p_j^A \ .
    \end{equation*}
    Summing up the contributions, the total cost of Algorithm~\ref{algo:golden wrr} can be bounded by
    \begin{equation*}
        \begin{aligned}
        \sum_{j=1}^{n} w_j C_j^{\text{WGR}^3} & = \sum_{j} \sum_k w_j \cdot c(k,j) && =  \sum_{j} \Big(\sum_{k\leq_A j}  w_j \cdot p_k^A + \sum_{k>_A j} w_k \cdot p_j^A\Big) \\
        &= \sum_{j} \Big(\sum_{k\leq_A j}  w_j \cdot p_k^A + \sum_{k<_A j} w_j \cdot p_k^A\Big) &&
        \leq 2 \cdot  \sum_{j} \sum_{k\leq_A j} w_j \cdot p_k^A \\
        &\leq 2 \cdot  \sum_{j} \sum_{k\leq_o j} w_j \cdot p_k^A && \leq 2\varphi \cdot\sum_{j} \sum_{k\leq_o j} w_j  \cdot p^*_k  = 2\varphi \cdot OPT \ .
        \end{aligned}
    \end{equation*}    
    The first inequality is obtained by rearranging the terms as outlined in the same way as in the proof of $2$-competitiveness of the $\text{WR}^3$, cf. \cite{kim2003non,pinedo2022scheduling}, as well. Because the order is w.r.t.~$p_j^A$, the same standard exchange argument as to prove the optimality of Smith's rule \cite{Smith1956WSPT} for $1 \st \st  \sum w_j C_j$ shows the second inequality. Finally, the transition from $p_j^A$ to $p_j^*$ yields another factor of $\varphi$ according to Lemma \ref{lemma: single job testing}. Tightness of the bound is achieved by the example for the unweighted case given in \cite{albers2021explorable}, which will also be considered in Example~\ref{example: golden weighted algorithm}.
\end{proof}

We now show Theorem~\ref{theorem: pcp algo 2varphi}, i.e., that the \textsc{Weighted} $\left(\varphi, 1 \right)$-\textsc{PCP} algorithm (i.e., Algorithm~\ref{algo: weighted pcp}) achieves a competitive ratio of $2\varphi$, being the test-preemptive version of the \textsc{Weighted Golden Round Robin Rule} (i.e., Algorithm~\ref{algo:golden wrr}).

\begin{proof}[Proof of Theorem~\ref{theorem: pcp algo 2varphi}]
    Using Lemma~\ref{lemma: non-decreasing completion times}, we obtain that both Algorithm~\ref{algo: weighted pcp} with $\alpha =\varphi, \beta=1$ and Algorithm~\ref{algo:golden wrr} complete tasks in the same order. Especially, jobs are completed in order of non-decreasing ratios $\nicefrac{p^A_j}{w_j}$ in both algorithms. With~$C_j^{A}$ denoting the completion time of job~$j$ in the first algorithm, we have that it is bounded at most by the sum of algorithmic processing times of jobs that are completed before $j$ and the testing times of jobs that are completed afterwards but are tested before $j$ is completed:
    \begin{equation}
        \label{eq: Weighted Algo against WRR}
            C_j^{A} 
            \leq \sum_{k\leq_{A} j}  p_k^A + \sum_{\substack{k>_{A} j\\ k \textup{ is tested} \\ w_j\cdot t_k\leq w_k \cdot p_j^A}} t_k 
            \leq \sum_{k\leq_{A} j} p_k^A + \sum_{k>_{A} j} \frac{w_k}{w_j} \cdot p_j^A 
            = C_j^{\text{WGR}^3} \ .
    \end{equation}
    Multiplying by $w_j$, summing up over $j$ and using Lemma~\ref{lemma: Golden WRR} we can conclude that Algorithm~\ref{algo: weighted pcp} with $\alpha =\varphi$ and $\beta=1$ is also $2\varphi$-competitive. 
\end{proof}
No tight example for this bound of the \textsc{Weighted} $\left(\varphi, 1 \right)$-\textsc{PCP} algorithm is currently known. The example that shows the tightness of the bound for Algorithm~\ref{algo:golden wrr} only gives a lower bound of~$2\varphi -1$, see Example~\ref{example: golden weighted algorithm}. The competitive ratio of $2\varphi$ for the \textsc{Weighted} $(\varphi,1)$-\textsc{PCP} algorithm can also be established differently, following a weighted extension of the analysis of \cite{gong2024approximation} or of~\cite{liu2023power}, see Theorem~\ref{theorem competitive ratio beta}, as well.

\begin{example}
    \label{example: golden weighted algorithm}
    Consider the same instance family as in the tight example for $\text{WGR}^3$ in \cite{albers2021explorable}. It consists of~$n$ jobs and parameters $u_j=p_j=w_j=1$ and $t_j=\nicefrac{1}{\varphi}$ for each job $j$.
    In the preemptive $\text{WGR}^3$ schedule, all testing tasks are completed by $\nicefrac{1}{\varphi}\cdot n$, while all tested-execution tasks finish by $\left(1+\nicefrac{1}{\varphi}\right)\cdot n$. The test-preemptive $\left(\varphi, 1 \right)$-\textsc{PCP} algorithm first schedules all testing tasks in a row, followed by all tested-execution tasks. This yields an objective value of $\nicefrac{n^2}{\varphi}+\nicefrac{\left(n+1\right)n}{2}$. The optimal offline algorithm does not test any job and schedules all jobs consecutively with an objective value of~$\nicefrac{\left(n+1\right)n}{2}$. Consequently, the competitive ratio of $\left(\varphi, 1 \right)$-\textsc{PCP} satisfies
    \begin{equation*}
        \frac{ALG}{OPT}= 1 + \frac{2\cdot n^2}{\varphi\cdot(n+1)\cdot n}\stackrel{n\longrightarrow \infty}{\longrightarrow} 1 + \frac{2}{\varphi} = 2\varphi -1 =  \sqrt{5}
        \geq 2.236 \ .
    \end{equation*}
    Thus, following this example, the exact competitive ratio of the \textsc{Weighted} $\left(\varphi, 1 \right)$-\textsc{PCP} algorithm lies between $2\varphi-1$ and $2\varphi$.
\end{example} 

\section{Delaying Testing and a Refined Analysis}
\label{sec: amortized analysis and delaying testing tasks}
In Subsection~\ref{subsec: best-known deterministic algo}, we extend the analysis of \cite{liu2023power} for the $\left(\alpha, \beta\right)$-\textsc{PCP} algorithm to incorporate weights (cf. Algorithm~\ref{algo: weighted pcp}) and give a lower bound of $2$ for any choice of parameters in Subsection~\ref{subsec: lower bound PCP}. In Subsection~\ref{subsec: best-known randomized algo}, we show how the results from~\cite{liu2023power} for the best-known randomized algorithm generalize to the weighted setting.

\subsection{A Deterministic Algorithm for $\sum w_jC_j$}
\label{subsec: best-known deterministic algo}
Table~\ref{table: beta results} gives an overview over the following results. The algorithms ($\alpha$, $\beta$)-\textsc{SORT} from \cite{albers2021explorable} and ($\alpha$, $\beta$)-\textsc{PCP} from \cite{liu2023power} are generalized to the broader setting of arbitrary job weights by scaling the priority values. Note that the only difference between those two algorithms is the priority value of a job that is reinserted. The lower bound of~$2$ for the first one stems from \cite{albers2021explorable}, while the proof for the second one is new. The upper bounds on the competitive ratios are shown by~\cite{liu2023power} for the unweighted case. We prove that they also hold for the more general case.

{\small
\begin{table}[htb]
    \centering
    \caption{Overview of results for the competitive ratios of the algorithms.}
        \label{table: beta results}
        \vspace{2ex}
        \begin{tabular}{c|c|c|c|c}
             Algorithm & Testing decision & Priority values & Lower bound & Upper bound \\
             \hline
            \makecell{ \textsc{Weighted} \\ ($\alpha$, $\beta$)-\textsc{SORT}}& \makecell{Test job iff \\ $\alpha t_j \leq u_j$}&$\frac{1}{w_j}\lbrace \beta t_j, u_j, p_j \rbrace$ & $2$ \cite{albers2021explorable} & \makecell{$1+ \sqrt{2}\leq 2.4143$ \\ for $(\sqrt{2},\sqrt{2})$}\\
            \hline
            \makecell{ \textsc{Weighted} \\ ($\alpha$, $\beta$)-\textsc{PCP}}& \makecell{Test job iff \\ $\alpha t_j \leq u_j$}&$\frac{1}{w_j} \lbrace \beta  t_j, u_j, t_j + p_j \rbrace$ & $2$ & \makecell{$2.3166$ \\ for $(\varphi,2.3166)$}
        \end{tabular}
\end{table}
}
Since we can rewrite the objective function as 
\begin{equation*}
    \sum_{j=1}^{n} w_j \cdot C_j^A = \sum_{j=1}^{n} \bigg(w_j \cdot c(j,j) + \Big( \sum_{k<_o j} \big(w_j \cdot c(k,j)+w_k \cdot c(j,k)\big)\Big)\bigg),
\end{equation*}
we can apply and generalize the idea based on pairwise contributions of jobs due to \cite{liu2023power} for the unweighted case and bound $w_j \cdot c(k,j)+w_k \cdot c(j,k)$ against its value in an optimal solution $w_j \cdot p^*_k$ in order to achieve an upper bound on the competitive ratio of the proposed algorithm. Lemma~\ref{lemma: test amortization} generalizes the analysis of~\cite{liu2023power} for $(\alpha,\beta)$-\textsc{PCP} to job-dependent weights. Their results for $(\alpha,\beta)$-\textsc{SORT} generalize analogously to job-dependent weights, cf. Appendix~\ref{seq: upper bound SORT}.

\begin{lemma} 
\label{lemma: test amortization}
Consider jobs $j\not = k$, and parameters $\alpha\geq 1$ and $\beta >0$. 
\\ If the \textsc{Weighted} $(\alpha,\beta)$-\textsc{PCP} algorithm does not test job $k$, it holds
\begin{equation*}
     w_j \cdot c(k,j)+w_k \cdot c(j,k) \leq w_j \cdot \Big(1+ \frac{1}{\beta}\Big) u_k \leq w_j \cdot \alpha\cdot \Big(1+ \frac{1}{\beta}\Big) \cdot p^*_k \ .
\end{equation*}
If the \textsc{Weighted} $(\alpha,\beta)$-\textsc{PCP} algorithm does test job $k$, it holds
\begin{equation*}
    \begin{aligned}
        w_j \cdot c(k,j)+w_k \cdot c(j,k) &\leq w_j \cdot \max\bigg\lbrace \beta \cdot t_k, 2t_k + p_k, (t_k+p_k)\Big(1+\frac{1}{\beta}\Big)\bigg\rbrace \\&\leq w_j \cdot \max\bigg\lbrace \beta, 2, 1+\frac{2}{\alpha}, \Big(1+\frac{1}{\alpha}\Big)\Big(1+\frac{1}{\beta}\Big)\bigg\rbrace \cdot p^*_k \ .
    \end{aligned}
\end{equation*}
\end{lemma}
\begin{proof}
    The proof extends the proof of~\cite{liu2023power} to job-dependent weights. We consider an exhaustive comparison tree based on the algorithm's actions. Similar reasoning has also been used in \cite{albers2021explorable,dogeas2024obligatory}. Assume that $k<_o j$ for a given pair of jobs $k,j$. In the optimal solution, we have mutual contributions of~$c^*(k,j)= p^*_k$ and~$c^*(j,k)=0$. 
    In each case, we find an upper bound on the ratio between the weighted sum of mutual contributions $w_j\cdot c(k,j)+w_k \cdot c(j,k)$ in the algorithm and~$w_j \cdot p^*_k$ in the optimum, depending on~$\alpha$ and~$\beta$.
    We extend the notation from~\cite{dogeas2024obligatory} and denote the testing task by $\tau_j$, the tested-execution task by $\pi_j$, and the untested-execution task of job~$j$ by $\upsilon_j$. An ordered list, such as $\left[ \upsilon_k, \tau_j,\pi_j\right]$ describes the testing and ordering decisions of the algorithm for jobs $k,j$ compactly. 
    
    First, we consider the cases in which the algorithm does not test job~$k$, see Table~\ref{table:no test on k}. To bound contributions of parts of job~$j$ against parts of job~$k$, we use inequality relations of priority values of tasks that come along with the algorithm's order of tasks in each case considered, independent of the tie-breaking rule used. Overall, it follows that $w_j\cdot c(k,j)+ w_k \cdot c(j,k) \leq w_j \cdot (1+ \nicefrac{1}{\beta}) u_k$ in this case. Bounding $u_k\leq\alpha p_k^*$ by Lemma~\ref{lemma: single job testing} yields the stated bound.

{\small 
        \begin{table}[htbp]
                \centering
                \caption{Case $p_k^A=u_k$\label{table:no test on k}}
                \vspace{2ex}
                \begin{tabular}{c|c|l|c}
                     &Order of Tasks & $w_j\cdot c(k,j)+w_k \cdot c(j,k)$ & Used Inequality \\
                     \hline
                     \multirow{2}{*}{$p_j^A=u_j$}
                     & $\left[ \upsilon_k,\upsilon_j\right]$ & $ w_j\cdot u_k$ & \\
                     &$\left[\upsilon_j,\upsilon_k\right]$  & $ w_k\cdot u_j \leq w_j\cdot u_k$ & $\frac{u_j}{w_j}\leq \frac{u_k}{w_k}$\\
                     \hline
                     \multirow{3}{*}{$p_j^A=t_j + p_j$} 
                     & $\left[ \upsilon_k,\tau_j, \pi_j\right]$ & $ w_j\cdot u_k$& \\
                     
                     & $\left[ \tau_j,\upsilon_k, \pi_j\right]$ & $ w_j\cdot u_k + w_k\cdot t_j \leq \big( 1+ \frac{1}{\beta} \big) w_j\cdot u_k$& $\frac{\beta \cdot t_j}{w_j}\leq\frac{u_k}{w_k}$\\
                     
                     & $\left[ \tau_j, \pi_j ,\upsilon_k\right]$ & $ w_k\cdot (t_j + p_j) \le w_j\cdot u_k $ & $\frac{t_j+p_j}{w_j}\leq \frac{u_k}{w_k}$
                \end{tabular}
            \end{table}
            }

 Secondly, we consider the cases when the algorithm tests job $k$, see Table~\ref{table:test on k}. Overall, we obtain $w_j \cdot c(k,j) + w_k \cdot c(j,k) \leq w_j \cdot \max\lbrace \beta \cdot t_k, 2t_k+p_k, \left(1+\nicefrac{1}{\beta}\right)\left(t_k+p_k\right)\rbrace$ in these cases, as stated above. We further bound the expressions with respect to the optimal running time $p^*_k$. First, we have~$\beta t_k \leq \beta p^*_k$, since~$p^*_k=\min\lbrace t_k+p_k, u_k\rbrace$ and $t_k\leq \alpha t_k \leq u_k$. Secondly, if $p^*_k=t_k+p_k$, we have $2t_k + p_k\leq 2t_k + 2p_k = 2p^*_k$. If, however, $p^*_k=u_k$, we have $2t_k + p_k = t_k + t_k+ p_k \leq \nicefrac{p^*_k}{\alpha} + \left(1+\nicefrac{1}{\alpha}\right)\cdot p^*_k= \left(1+\nicefrac{2}{\alpha}\right)p^*_k$ by Lemma~\ref{lemma: single job testing} and $t_k\leq \nicefrac{u_k}{\alpha}$ since $k$ was tested.
Lastly, $\left(1+\nicefrac{1}{\beta}\right)\left(t_k+p_k\right)\leq \left(1+\nicefrac{1}{\beta}\right)\cdot\left(1+\nicefrac{1}{\alpha}\right)p^*_k$ by Lemma~\ref{lemma: single job testing}. These calculations prove the second inequality in the statement of the Lemma.

{\small
    \begin{table}[htbp]
        \centering
        \caption{Case $p_k^A= t_k + p_k$\label{table:test on k}}
        \vspace{2ex}
        \begin{tabular}{c|c|l|c}
             &Order of Tasks & $w_j\cdot c(k,j)+w_k \cdot c(j,k)$ & Used Inequality \\
             \hline
            \multirow{3}{*}{\rotatebox[origin=c]{90}{$p_j^A=u_j$}}
            &$\left[\tau_k, \pi_k, \upsilon_j \right]$ & $w_j\cdot (t_k+p_k)$ & \\
            
            &$\left[ \tau_k,\upsilon_j,\pi_k\right]$ & $ w_j\cdot t_k + w_k\cdot u_j \leq w_j\cdot (2t_k+ p_k)$ & $\frac{u_j}{w_j}\leq\frac{t_k+p_k}{w_k}$\\
            
            &$\left[ \upsilon_j, \tau_k, \pi_k \right]$ & $ w_k\cdot u_j \leq \beta\cdot w_j \cdot t_k$ & $\frac{u_j}{w_j}\leq \frac{\beta \cdot t_k}{w_k}$\\
            \hline
            \multirow{6}{*}{\rotatebox[origin=c]{90}{$p_j^A=t_j + p_j$}}
                 &$\left[\tau_k, \pi_k, \tau_j, \pi_j \right]$ & $ w_j\cdot (t_k+p_k)$& \\
                 
                 &$\left[\tau_k,  \tau_j, \pi_k, \pi_j \right]$ & $ w_j\cdot (t_k+p_k) + w_k\cdot t_j \leq w_j\cdot \big(1+\frac{1}{\beta}\big)\left(t_k+p_k\right)$ & $\frac{\beta \cdot t_j}{w_j}\leq\frac{t_k+p_k}{w_k}$\\

                 &$\left[\tau_k,  \tau_j, \pi_j, \pi_k \right]$ &$ w_j\cdot t_k + w_k\cdot (t_j+p_j) \leq w_j\cdot (2t_k+ p_k)$ & $\frac{t_j+p_j}{w_j}\leq\frac{t_k+p_k}{w_k}$\\
                 
                 &$\left[\tau_j, \tau_k, \pi_k, \pi_j \right]$ & $ w_j\cdot (t_k+p_k) + w_k\cdot t_j \leq w_j\cdot \left(2t_k+p_k\right)$ & $\frac{\beta \cdot t_j}{w_j}\leq\frac{\beta \cdot t_k}{w_k}$ \\
                 
                 &$\left[\tau_j, \tau_k, \pi_j, \pi_k \right]$ & $ w_j\cdot t_k + w_k\cdot (t_j+p_j) \leq w_j\cdot (2t_k+ p_k)$& $\frac{t_j+p_j}{w_j}\leq\frac{t_k+p_k}{w_k}$\\
                 
                 &$\left[\tau_j, \pi_j , \tau_k, \pi_k \right]$ & $ w_k\cdot(t_j+p_j) \leq \beta\cdot w_j \cdot t_k$& $\frac{t_j+p_j}{w_j}\leq \frac{\beta \cdot t_k}{w_k}$
            \end{tabular}
        \end{table}
        }
\end{proof}
\noindent Intuitively, two different trade-offs in terms of uncertainty are incorporated in the bounds of Lemma~\ref{lemma: test amortization}. The parameter $\alpha$ balances the risks of testing jobs that turn out to be long and not testing jobs that are actually short, cf. Lemma~\ref{lemma: single job testing}. The parameter $\beta$ balances the risk of delaying many jobs when one tests jobs that turn out to be long early and delaying tests of actually short jobs by too much. The bounds from Lemma~\ref{lemma: test amortization} will now be applied to derive an upper bound on the competitive ratio of the \textsc{Weighted} $(\alpha,\beta)$-\textsc{PCP} as in~\cite{liu2023power}.
\begin{theorem}
\label{theorem competitive ratio beta}
    The competitive ratio of the \textsc{Weighted} $(\alpha,\beta)$-\textsc{PCP} algorithm with $\alpha\geq 1$ and $\beta>0$ is at most 
    \begin{equation*}
    Q(\alpha,\beta) := \max\bigg\lbrace \alpha\Big(1+\frac{1}{\beta}\Big), \beta, 2, 1+\frac{2}{\alpha}, \Big(1+\frac{1}{\alpha}\Big)\Big(1+\frac{1}{\beta}\Big)\bigg\rbrace \ .
    \end{equation*}
    The optimal choice of the parameters to minimize this term is $\hat{\alpha}=\varphi$ and $\hat{\beta} = \nicefrac{\left(\varphi+\sqrt{\varphi\left(\varphi+ 4\right)}\right)}{2}$, resulting in a competitive ratio of at most $\hat{\beta} \leq 2.3166$.
\end{theorem}
\begin{proof}
    \begin{equation}
    \label{eq: competitive ratio PCP}
    \begin{aligned}
        \sum_{j=1}^{n} w_j\cdot C_j^A &= \sum_{j=1}^{n} \Big( w_j \cdot c(j,j) +  \sum_{k<_o j} \big( w_j \cdot c(k,j)+w_k \cdot c(j,k)\big)\Big)
        \\&\leq \sum_{j=1}^{n} \Big(w_j\cdot p_j^A + \sum_{k<_o j} \big(w_j \cdot Q(\alpha,\beta)\cdot p^*_k\big)\Big)
        \\&\leq Q(\alpha,\beta)\cdot \sum_{j=1}^{n} \sum_{k\leq_o j} w_j\cdot  p^*_k = Q(\alpha,\beta)\cdot OPT
    \end{aligned}
    \end{equation} 
    The first inequality to bound the pairwise contributions of jobs follows from Lemma~\ref{lemma: test amortization}. The second inequality holds, since $\big(1+\nicefrac{1}{\beta}\big)\geq 1$ for $\beta>0$, and, thus, $p_j^A\leq\max\lbrace \alpha, 1+\nicefrac{1}{\alpha}\rbrace\cdot p^*_j \leq Q(\alpha,\beta) \cdot p^*_j$ according to Lemma~\ref{lemma: single job testing} and the definition of $Q(\alpha,\beta)$ in the statement of Theorem~\ref{theorem competitive ratio beta}. Minimizing this upper bound over $\alpha$ and $\beta$ yields $\hat{\alpha}=\varphi$ and $\hat{\beta} = \nicefrac{\left(\varphi+\sqrt{\varphi\left(\varphi+ 4\right)}\right)}{2}$ as optimal values, which is inherited from the unweighted case \cite{liu2023power}: the terms $\alpha\big(1+\nicefrac{1}{\beta}\big)$ and $\big(1+\nicefrac{1}{\alpha}\big)\big(1+\nicefrac{1}{\beta}\big)$ intersect at $\alpha=\varphi$ for any fixed $\beta>0$ as in Lemma~\ref{lemma: single job testing}. Solving $\beta=\varphi\big(1+\nicefrac{1}{\beta}\big)$ yields the stated formula for $\hat{\beta}$. The three terms are equal for this parameter combination, and at least one is strictly larger for any other parameter combination. In addition, we have that $\max\lbrace 2,1+\nicefrac{2}{\varphi} \rbrace=\max\lbrace 2,\sqrt{5} \rbrace < \hat{\beta}$ concluding the proof. 
\end{proof}

\subsection{A Lower Bound for the (\textsc{Weighted}) $(\alpha,\beta)$-PCP Algorithm}
\label{subsec: lower bound PCP}
For $(\alpha,\beta)$-\textsc{SORT} (cf. Table~\ref{table: beta results}), a lower bound of $2$ was established by \cite{albers2021explorable}. By following a similar proof structure with adapted instance families, we derive the same lower bound for~$(\alpha,\beta)$-\textsc{PCP}. In particular, for the choice of parameters $\hat{\alpha}=\varphi$ and $\hat{\beta} = \nicefrac{\left(\varphi+\sqrt{\varphi\left(\varphi+ 4\right)}\right)}{2}$ as in Theorem~\ref{theorem competitive ratio beta}, Example~\ref{example: golden weighted algorithm} can be used to obtain a lower bound of $2\varphi -1 > 2.2360$ as for $\left(\varphi, 1 \right)$-\textsc{PCP}.

\begin{lemma}
    \label{lemma: lower bounds}
    The (\textsc{Weighted}) $(\alpha,\beta)$-\textsc{PCP} algorithm cannot be better than $2$-competitive for any combination of $\alpha\geq 1$ and $\beta \geq 0$ for scheduling with testing on a single machine to minimize the weighted total completion time.
\end{lemma}
\begin{proof}
     We distinguish three parameter combinations of $\alpha\geq 1$ and~$\beta\geq 0$ and construct instance families for each case. The constructions are similar to those for the $(\alpha,\beta)$-\textsc{SORT} in~\cite{albers2021explorable}, but adapted to the priority structure of the $(\alpha,\beta)$-\textsc{PCP}.
    \begin{enumerate}
        \item $\alpha > 2$: Following Lemma~\ref{lemma: single job testing}, consider the single job instance with $$t_1 = 1, u_1=\alpha-\varepsilon, p_1=0, w_1=1\ .$$ 
            $(\alpha,\beta)$-\textsc{PCP} does not test the job, whereas the optimal offline algorithm does, yielding a ratio of~$\nicefrac{ALG}{OPT}=\alpha-\varepsilon>2$ for small $\varepsilon>0$.
        \item $1 \leq \alpha \leq 2$ and $0 \leq \beta < 1+ \alpha$:
            In this case, we design instances as in Example~\ref{example: golden weighted algorithm} where $(\alpha,\beta)$-\textsc{PCP} tests jobs that should not be tested, and in addition, schedules all tests before any execution, causing high delay costs. Consider the family of instances consisting of $n$ jobs with
            $$t_j=1, u_j= \alpha, p_j=\alpha, w_j=1 \textup{ for } j=1,\ldots,n\ .$$ 
            
            Note that $(\alpha,\beta)$-\textsc{PCP} tests all jobs, because $\nicefrac{u_j}{t_j}\geq \alpha$ for each job~$j$, and schedules all tests first, since~$\beta t_j = \beta  < 1+ \alpha  = t_k + p_k$ for all jobs~$j,k$. Therefore, the algorithm's objective value amounts to $ALG=n^2 + \nicefrac{\alpha n(n+1)}{2}$.
            
            The optimal offline algorithm does not test any job, and obtains an objective value of $OPT=\nicefrac{\alpha n(n+1)}{2}$. For $n\rightarrow \infty$, we get a competitive ratio of $\nicefrac{ALG}{OPT}\rightarrow 1+\nicefrac{2}{\alpha}\geq 2$, since $1 \leq \alpha \leq 2$.
        \item $1 \leq \alpha \leq 2$ and $\beta \geq 1+ \alpha$:
        Consider a family of instances with two types of jobs and small $\varepsilon>0$:
        \begin{equation*}
            \begin{aligned}
                &\textup{(\emph{long}) } &&t_j=1-\varepsilon, u_j= \beta -1, p_j=\beta-1, w_j=1 &&\textup{ for } j=1,\ldots,n,\\
                 &\textup{(\emph{short}) } &&t_j=1, \hspace{4ex}u_j= \alpha, \hspace{4ex}p_j=0, \hspace{4.5ex}w_j=1 &&\textup{ for } j=n+1,\ldots,2n\ .
            \end{aligned}
        \end{equation*}
        In these instances, several incorrect decisions of the $(\alpha,\beta)$-\textsc{PCP} accumulate, especially given the relatively large testing delay. 
        Both $(\alpha,\beta)$-\textsc{PCP} and the optimal offline algorithm test short jobs by construction. Long jobs are not tested by the optimal offline algorithm, but $(\alpha,\beta)$-\textsc{PCP} does, since $\alpha t_j\leq \alpha \leq \beta -1 = u_j$ for $j=1,\ldots, n$.
        The $(\alpha,\beta)$-\textsc{PCP} algorithm schedules the tests of the long jobs first (priority value $\beta(1-\varepsilon)$), then it executes them (priority value $1-\varepsilon+\beta-1=\beta-\varepsilon$), and afterwards it tests and executes the short jobs directly after their respective tests (priority values~$\beta$ and~$1$). The high testing delay causes long jobs to be scheduled even before short jobs are tested. Including the parameter~$\varepsilon$ enforces the wrong ordering, but it can be neglected if ties are broken as described.
        
        The optimal offline algorithm tests and executes all short jobs first, then executes the long jobs untested. Overall, it follows
        \begin{equation*}
            \begin{aligned}
                \frac{ALG}{OPT} &= \frac{(1-\varepsilon)n^2 + (\beta-1)\frac{n(n+1)}{2} + (\beta-\varepsilon)n^2 + \frac{n(n+1)}{2}}{\frac{n(n+1)}{2} + n^2 + (\beta-1)\frac{n(n+1)}{2}} \\
                &\stackrel{n\rightarrow\infty}{\longrightarrow} \frac{\frac{3\beta}{2}+1-2\varepsilon}{\frac{\beta}{2}+1}\stackrel{\varepsilon\rightarrow 0}{\longrightarrow} 3- \frac{4}{\beta+2}\stackrel{\beta\geq 2}{\geq} 2\ .
            \end{aligned}
        \end{equation*}   
    \end{enumerate} 
\end{proof}
    Note that for the lower bound of the (\textsc{Weighted}) $(\alpha,\beta)$-\textsc{SORT} one may replace $0\leq \beta < 1+ \alpha$ by $0\leq \beta < \alpha$ in case~$2$ of the previous proof. Consequently, we may replace $\beta \geq 1+ \alpha$ by $\beta \geq \alpha$ and $u_j=p_j=\beta-1$ by $u_j=p_j=\beta-\varepsilon$ for the long jobs $j=1,2,\ldots, n$ in case~$3$, similar as in \cite{albers2021explorable}.

Interestingly, this lower bound matches the competitive ratio of $2$ of the \textsc{Threshold} algorithm from~\cite{durr2020adversarial} for $w_j=t_j=1$. This algorithm coincides with the $(2,2)$-\textsc{SORT} algorithm on those special instances. Therefore, the above results show that even with unit weights and testing times, $(\alpha,\beta)$-\textsc{SORT} and $(\alpha,\beta)$-\textsc{PCP} cannot achieve a better competitive ratio for any choice of parameters $\alpha\geq 1$ and $\beta \geq 0$.

\subsection{A Randomized Algorithm for $\sum w_jC_j$}
\label{subsec: best-known randomized algo}
An idea previously used when designing randomized algorithms for scheduling with testing is to randomize the testing decisions in the beginning and then use a deterministic algorithm for the choice of the next task to be scheduled, see~\cite{albers2021explorable,gong2024approximation,liu2023power}. One typically uses a non-decreasing function $f: \mathbb{R}_{> 0}\rightarrow \left[ 0,1\right]$ that maps the ratio $r_j=\nicefrac{u_j}{t_j}$ to a probability $f(r_j)$ for testing job~$j$. The randomized algorithm from~\cite{liu2023power} for the unweighted case uses the following function:
\begin{equation}
    \label{eq: prob function}
    f(x) : = \left \{ \begin{array}{cc}
        0, & x < 1, \\
        \frac{x^2 - x}{x^2 - \frac{4}{3} x +1}, & 1\leq x \leq 3,  \\ 
         1, & x \geq 3 \ .
    \end{array}\right.
\end{equation}
They give a worst-case analysis that provides such a function based on a general~$\beta$ and set $\beta=2$ to achieve the best-possible bound within their analysis. For better readability, we only use the optimized version here to show how it can be generalized to the weighted case. We define the algorithm \textsc{Weighted} $(f, \beta)$-\textsc{PCP}, as the one that tests job~$j$ with probability $ f(r_j)$ and then always schedules the task with minimal priority value (either $\nicefrac{1}{w_j}\cdot\beta \cdot t_j$ or $\nicefrac{1}{w_j}\cdot u_j$ or $\nicefrac{1}{w_j}\cdot \left(t_j + p_j\right)$ depending on the task's type) among the available tasks.

\begin{theorem}
    \label{theorem: randomized bound single machine} 
    The randomized algorithm \textsc{Weighted} $(f, 2)$-\textsc{PCP} achieves a competitive ratio of at most $\nicefrac{3(7+3\sqrt{6})}{20}\leq 2.1523$.
\end{theorem}
\begin{proof}
    Using the randomized algorithm \textsc{Weighted} $(f, 2)$-\textsc{PCP} and following the proof ideas and structure for the unweighted case \cite{liu2023power}, the expected weighted total completion time can be bounded by
    \begin{equation*}
        \begin{aligned}
            E\Big[ \sum_j w_j \cdot C_j^A\Big] = \sum_{j} \bigg(E\left[w_j  \cdot c(j,j)\right] + E\Big[ \sum_{k<_o j} w_j \cdot c(k,j)+w_k \cdot c(j,k)\Big] \bigg)\\
            \leq \sum_{j}w_j \cdot \sum_{k\leq_o j} \underbrace{\frac{3}{2} u_k \left( 1- f\left(r_k\right)\right) + \max\Big\lbrace 2t_k+p_k,\frac{3}{2}\left(t_k+p_k\right)\Big\rbrace f\left(r_k\right)}_{=: h(r_k)},
        \end{aligned}
    \end{equation*}
    where the inequality follows from Lemma~\ref{lemma: test amortization} for $\beta=2$ for the second expectation, and since the first expectation $E\left[w_j \cdot c(j,j)\right] = w_j u_j \left( 1- f\left(r_j\right)\right) + w_j\left(t_j+p_j\right)f\left(r_j\right)$ is clearly also bounded by $w_j\cdot h(r_j)$. Overall, $h(r_k)$ can be bounded based on a case distinction on~$r_k$ and on $p_k^*\in\lbrace u_k, t_k + p_k\rbrace$ as follows, similar to \cite{liu2023power}.
    Firstly, for $r_k< 1$, we have $f(r_k)=0$ and $h(r_k)=\nicefrac{3}{2}\cdot u_k=\nicefrac{3}{2} \cdot p_k^*$. 
    Secondly, for~$r_k\geq 3$ using $f\left(r_k\right)=1$, we obtain $h(r_k)\leq 2\cdot\left(t_k+p_k\right)$ and
    \begin{equation*}
        \begin{aligned}
        h(r_k) = \max\Big\lbrace 2\frac{u_k}{r_k} + p_k, \frac{3}{2}\big(\frac{u_k}{r_k} + p_k\big) \Big\rbrace \leq \max\Big\lbrace u_k\big( \frac{2}{3} + 1\big), \frac{3}{2}u_k\big(\frac{1}{3}+1\big) \Big\rbrace  \leq 2\cdot u_k ,  
        \end{aligned}
    \end{equation*}
    hence $h(r_k)\leq 2\cdot p_k^*$.
    Finally, consider $1\leq r_k\leq 3$. We obtain

    \begin{equation*}
        \begin{aligned}
            \frac{h(r_k)}{t_k+p_k}\leq \frac{3}{2}r_k\left( 1-f(r_k)\right) + 2 f(r_k) = \frac{3}{2}r_k + f(r_k)\Big(2-\frac{3}{2}r_k\Big)
        \end{aligned}
    \end{equation*}
    and using $r_k\geq 1$ in the second step of the next calculation
   \begin{equation*}
        \begin{aligned}
            \frac{h(r_k)}{u_k}&\leq \frac{3}{2}\left( 1-f(r_k)\right) +\max\Big\lbrace \frac{2}{r_k} +1, \frac{3}{2}\big(\frac{1}{r_k}+1\big)\Big\rbrace f(r_k) \\
            &= \frac{3}{2}\left( 1-f(r_k)\right) + \frac{3}{2}\big(1+\frac{1}{r_k}\big)f(r_k) = \frac{3}{2} + f(r_k)\frac{3}{2r_k}\ .
        \end{aligned}
    \end{equation*}
    Therefore, we have 
   \begin{equation*}
       \begin{aligned}
        h(r_k)\leq\max\Big\lbrace \frac{3}{2}r_k + f(r_k)\big(2-\frac{3}{2}r_k\big) ,\frac{3}{2} + f(r_k)\frac{3}{2r_k} \Big\rbrace\cdot p_k^*\ .
       \end{aligned}
   \end{equation*}
    For each $r_k\in\left[1,3\right]$, $f(r_k)\in \left[0,1\right]$ as defined in Equation~\ref{eq: prob function} is derived by minimizing the previous factor for a general testing probability. Observing that the second term in the maximum grows strictly faster in the testing probability than the first term, one considers the intersection, and for $f(r_k)$, the two terms within the maximum are equal to $\nicefrac{(9r_k^2-3r_k)}{(6r_k^2-8r_k+6)}$. This expression is largest for $r_k= 1 + \sqrt{\nicefrac{2}{3}}$ yielding $h(r_k) \leq \nicefrac{3(7+3\sqrt{6})}{20} \cdot p^*_k \leq  2.1523 \cdot p^*_k$.
\end{proof}

\section{Parallel Machines}
\label{section: Parallel Machines}
Combining ideas from single-machine algorithms with list scheduling yields algorithms for scheduling on identical parallel machines. One essentially keeps the same task order as in the single-machine case (i.e., that of the \textsc{Weighted} $(\alpha,\beta)$-PCP algorithm) and schedules the next available task once a machine idles, albeit respecting that a tested-execution task is only available after the completion of a testing task of the same job, similar to \cite{gong2024approximation,liu2024parallel} for minimizing the total completion time, see Algorithm~\ref{algo: parallel machines}. It is important to note that a tested-execution task is placed at the relative position in the priority list it would have had in the single-machine case and is not simply reinserted with its priority value. The overall analysis is similar to proving a performance ratio of $\nicefrac{3}{2}$ for \textsc{WSPT} on identical parallel machines in the offline setting based on the lower bounds from~\cite{eastman1964bounds}, cf.~\cite{mohring1999approximation}. For the setting considered here, we obtain a deterministic algorithm with a competitive ratio of~$2.7763$, replicating the performance guarantees in~\cite{liu2024parallel} where the unweighted case was studied. Furthermore, we obtain a randomized $2.5110$-competitive algorithm improving upon the so far best-known bound for randomized algorithms on parallel machines of $2.8307$ for the unweighted case from~\cite{gong2024approximation}.

\begin{algorithm}[htb]
    \caption{Weighted $(\alpha,\beta)$-PCP on Identical Parallel Machines}
    \label{algo: parallel machines}
    \KwIn{Jobs $\calJ$, a threshold $\alpha\geq 1$, a delay-factor $\beta>0$, $m$ identical parallel machines.}
    \KwOut{A feasible test-preemptive schedule on $m$ identical parallel machines.}
    \BlankLine
    Jobs to be tested: $\calJ_T := \lbrace j \in \calJ \st \alpha t_j \leq  u_j \rbrace$; Jobs left untested $\calJ_U := \calJ\setminus\calJ_T$;
    \\$L:=$ Empty priority list (keeping the single-machine order of unfinished tasks as in Algorithm~\ref{algo: weighted pcp});  
    \\Add tasks for jobs from $\calJ_T$ with $\sigma_j =\nicefrac{\beta t_j}{w_j}$, from $\calJ_U$ with $\sigma_j =\nicefrac{u_j}{w_j}$ to $L$; 
    \\Initialize $\tau:=0$ as the current time; 
    \BlankLine
    \While{$L\neq\emptyset$ or some machine is busy} 
    {
        Let $m'\leq m$ be the number of machines that are idle from $\tau$ onwards; 
        \\Schedule the first $\min\lbrace m', |L|\rbrace$ tasks in $L$ on those machines starting at time $\tau$ and remove them from $L$;
        \\Set $\tau$ to the next completion time of any running task;
        \\\For{each testing task of a job $j$ that is completed at $\tau$}{Reinsert a task into $L$ using priority value $\sigma_j=\max\lbrace\nicefrac{\beta t_j}{w_j}, \nicefrac{\left(t_j + p_j\right)}{w_j}\rbrace$, i.e., at the position it would occupy in a single-machine order, tie-breaking arbitrarily.
        }
    }
\end{algorithm}
Preserving the single-machine order for the priority list in Algorithm~\ref{algo: parallel machines} allows for a simple proof of Lemma~\ref{lemma: from single to parallel machines}.
\begin{lemma}
    \label{lemma: from single to parallel machines}
    The completion time~$C_j$ of a job~$j$ in the schedule constructed by the list scheduling approach in Algorithm~\ref{algo: parallel machines} can be bounded by:
	\begin{equation}
        \label{eq: C_j from single to parallel machines}
		C_j \leq \frac{1}{m} \sum_{k\not= j} c(k,j) + p_j^A,
	\end{equation}
    where $c(k,j)$ is the contribution of job~$k$ to the completion time of job $j$ in the single-machine schedule constructed by the \textsc{Weighted} $(\alpha,\beta)$-PCP algorithm with the same tie-breaking rule.
\end{lemma}
\begin{proof} 
 The key observation of the proof is that list scheduling assigns a task to the currently least-loaded machine or a tested-execution task immediately after the completion of its testing task. In the following, let Algorithm~\ref{algo: parallel machines} use the same (arbitrary, but fixed) tie-breaking rule as Algorithm~\ref{algo: weighted pcp} when simulating the single-machine order.
 
 For the proof, we need the notion of a \emph{violation of the single-machine order}, c.f. \cite{gong2024approximation}, which occurs for a task that is only started on one of the machines after another task that appears later in the single-machine order. This can only occur for tested-execution tasks that are not initially available. While a testing task of job~$j$ is ongoing, another task that appears later in the single-machine order than the tested-execution task of job~$j$ may be scheduled on another machine. Let us distinguish three cases depending on the testing decision for a job and whether violations of the single-machine order occur.
 
 If job $j$ is not tested, the corresponding task is initially available. Therefore, only tasks preceding~$j$ in the single-machine order may contribute to its starting time. Since $j$ is put on the least-loaded machine at that point in time, its starting time is at most $\nicefrac{1}{m} \sum_{k\not= j} c(k,j)$. Since job~$j$ is scheduled non-preemptively, its completion time is at most $\nicefrac{1}{m} \sum_{k\not= j} c(k,j) + p_j^A$. 
 
 Consider a job $j$ that is tested, and for which no violation of the single-machine order occurs, i.e., no task that appears after the tested-execution task of~$j$ in the single-machine order is started earlier in the parallel-machine schedule. We distinguish two subcases. If the tested-execution task directly follows the testing, the completion time is bounded by the starting time of the testing task plus $p_j^A$. Applying the argument of the previous case to the testing task yields the bound. Otherwise we obtain $C_j\leq \nicefrac{1}{m} \left(\sum_{k\not= j} c(k,j) + t_j\right) + p_j$. Using $p_j^A=t_j+p_j$ and $m\geq 1$ yields the stated bound.

Finally, consider a set of tested-execution tasks whose testing tasks complete at time $\tau'$ and for which violations of the single-machine order occur. By construction, if a task appearing later in the single-machine order has already been scheduled before, no other task with a smaller priority can be present in the priority list $L$ at time $\tau'$, since it would have been scheduled before this violating task otherwise. Thus, Algorithm~\ref{algo: parallel machines} immediately schedules all these tasks again. Then, the bound holds since one only needs to consider the contributions of tasks before the testing of $j$ as in the previous cases. 
\end{proof}

 For the underlying offline problem, and more generally for the preemptive identical-parallel-machine scheduling problem $P \st pmtn \st \sum w_j C_j$, there exist optimal schedules without preemption \cite{brucker2007scheduling}. Furthermore, the optimal offline algorithm uses the optimal processing time $p_j^*$ for each job~$j$.
 Since we do not have the concrete structure of the optimal solution to the non-preemptive problem $P \st \st \sum w_j C_j$ at hand, we use the following lower bound of \cite{eastman1964bounds} to derive an upper bound on the competitive ratio:
\begin{equation}    
    \label{eq: Eastman}
	OPT_m \geq \max \bigg\lbrace  \sum_j w_j p^*_j, \frac{1}{m} OPT_1 + \Big(\frac{1}{2}-\frac{1}{2m} \Big) \sum_j w_j p^*_j  \bigg\rbrace ,
\end{equation}
where $OPT_m$ is the optimal objective value on~$m$ machines and $OPT_1$ on a single machine. Equations~\eqref{eq: C_j from single to parallel machines} and~\eqref{eq: Eastman} were similarly used for the unweighted case in~\cite{gong2024approximation,liu2024parallel} and in many other scheduling contexts before.

\begin{theorem}
    \label{theorem: parallel machines proof}
	\textsc{Weighted} $\left(\alpha,\beta\right)$-\textsc{PCP}  with parameters as in Theorem~\ref{theorem competitive ratio beta} combined with list scheduling as in Algorithm~\ref{algo: parallel machines} yields a competitive ratio of at most $2.3166$ for $m=1,2$ and $2.7763 - \nicefrac{1.1582}{m}$ for~$m\geq 3$.
\end{theorem}
\begin{proof}
    We first provide a proof using general parameters $\alpha\geq 1$ and $\beta>0$ and then plug in the values as in Theorem~\ref{theorem competitive ratio beta}, $\hat{\alpha}$ and $\hat{\beta} = \nicefrac{\left(\varphi+\sqrt{\varphi\left(\varphi+ 4\right)}\right)}{2}$. For this proof, we apply the per-job completion time bound from Equation~\eqref{eq: C_j from single to parallel machines} in Lemma~\ref{lemma: from single to parallel machines} inherited from the list-scheduling framework, then we use the single-machine pairwise-contribution bound of jobs as established by Equation~\eqref{eq: competitive ratio PCP} in Theorem~\ref{theorem competitive ratio beta} and the job-wise bound from the testing decision as in Lemma~\ref{lemma: single job testing}. We rewrite the resulting terms to apply the lower bound from Equation~\eqref{eq: Eastman} twice, yielding terms parametrized by the number of machines~$m$ and parameters $\alpha$ and $\beta$.
	\begin{equation*}
		\begin{aligned}
			\sum_j w_j C_j 
            & \stackrel{\eqref{eq: C_j from single to parallel machines}}{\leq} \frac{1}{m} \sum_j w_j \sum_{k\not= j} c(k,j) + \sum_j w_j p_j^A \\
            & \hspace{-1.5ex}\stackrel{\eqref{eq: competitive ratio PCP},~\ref{lemma: single job testing}}{\leq} \frac{1}{m} Q(\alpha,\beta)\sum_j w_j \sum_{k\not= j} c^*(k,j) + \max\lbrace \alpha,1+\nicefrac{1}{\alpha}\rbrace \sum_j w_j p_j^* \\
            &= \frac{1}{m} Q(\alpha,\beta) OPT_1 + \Big(\max\lbrace \alpha,1+\nicefrac{1}{\alpha}\rbrace-\frac{Q(\alpha,\beta)}{m}\Big)\sum_j w_j p_j^* \\            
            & = Q(\alpha,\beta)\bigg(\frac{1}{m} OPT_1 + \Big(\frac{1}{2}-\frac{1}{2m}\Big)\sum_j w_j p^*_j\bigg) \\
            &\hspace{1ex} + \bigg(\max\lbrace \alpha,1+\nicefrac{1}{\alpha}\rbrace -Q(\alpha,\beta)\Big(\frac{1}{m}+\frac{1}{2}-\frac{1}{2m}\Big)\bigg) \sum_j w_j p_j^* \\
            &\stackrel{\eqref{eq: Eastman}}{\leq} \left \{ 
                \begin{array}{ll}
                    OPT_m \cdot Q(\alpha,\beta), & \textup{if }\max\lbrace \alpha,1+\nicefrac{1}{\alpha}\rbrace \leq Q(\alpha,\beta)\big(\frac{m+1}{2m}\big), \\
                    OPT_m \cdot \left( \left(\frac{1}{2}-\frac{1}{2m}\right)Q(\alpha,\beta)+\max\lbrace \alpha,1+\nicefrac{1}{\alpha}\rbrace\right), & \textup{else} .
                \end{array}\right. 
		\end{aligned}
	\end{equation*}
    For the single-machine parameters $\hat{\alpha}=\varphi$ and $\hat{\beta}=\nicefrac{\left(\varphi+\sqrt{\varphi\left(\varphi+ 4\right)}\right)}{2}$ from Theorem~\ref{theorem competitive ratio beta}, we obtain the single-machine bound~$Q(\hat{\alpha},\hat{\beta})$ for $m=1,2$ from the first case above. For $m\geq 3$, the stated bound follows from the second case above and since $-\nicefrac{Q(\hat{\alpha},\hat{\beta})}{2}<-1.1582$. It is monotonically increasing in~$m$ and converges to $\nicefrac{Q(\hat{\alpha},\hat{\beta})}{2}+\varphi \leq 2.7763$ for $m\rightarrow\infty$.
    Note that when using exactly the above analysis for the \textsc{Weighted} $\left(\alpha,\beta\right)$-\textsc{PCP}, no parameter combination that is constant for all $m$ can asymptotically be better than~$\hat{\alpha}$ and~$\hat{\beta}$, since $\nicefrac{Q(\alpha,\beta)}{2}+ \max\lbrace \alpha,1+\nicefrac{1}{\alpha}\rbrace\geq \nicefrac{Q(\hat{\alpha},\hat{\beta})}{2} + \varphi$ for all $\alpha\geq 1$ and~$\beta>0$ according to Theorem~\ref{theorem competitive ratio beta} and Lemma~\ref{lemma: single job testing}. 
    
    Note that the bound is slightly improved for every fixed $m\geq 2$ compared to the conference version~\cite{buld2025scheduling} proving a competitive ratio of $\frac{1}{2}Q(\hat{\alpha},\hat{\beta})+\varphi + \nicefrac{1}{m} \left( \nicefrac{Q(\hat{\alpha},\hat{\beta})}{2} - \varphi \right) \leq 2.7763 - \nicefrac{0.4597}{m}$. The difference lies in the second inequality above. In this version, we bound the total pairwise contributions $\sum_j w_j \sum_{k\not= j} c(k,j)$ using the factor $Q(\alpha,\beta)$ instead of the algorithmic cost $\sum_j w_j \sum_{k} c(k,j)$ on a single machine, compare with the proof of Theorem~\ref{theorem competitive ratio beta}.
 \end{proof}

Following essentially the same steps, we can use this analysis to improve upon the best-known bound for randomized algorithms on parallel machines, $2.8307$, from \cite{gong2024approximation} for the unweighted case, even for job-dependent weights. Instead of a threshold parameter $\alpha$, we use the probability function~$f$ of Subsection~\ref{subsec: best-known randomized algo}.

\begin{theorem}
    \textsc{Weighted} $(f,2)$-PCP combined with list scheduling as in Algorithm~\ref{algo: parallel machines} has a competitive ratio of at most $\nicefrac{3(7+3\sqrt{6})}{20}\leq 2.1523$ for $m\leq 3$ and $\left( 7 - \nicefrac{3}{m}\right) \nicefrac{\left( 7 + 3\sqrt{6}\right)}{40}\leq 2.5110-\nicefrac{1.0761}{m}$ for $m\geq 4$.
\end{theorem}
\begin{proof}
    The proof follows the same framework as in Theorem~\ref{theorem: parallel machines proof}. We replace $\sum_j w_j C_j $ by $E\big[\sum_j w_j C_j \big]$ and $p_j^A$ by $E\big[p_j^A\big]$, apply the bound from Equation~\eqref{eq: C_j from single to parallel machines} on $C_j$ inherited from the list-scheduling framework in each realization, and make use of the linearity of expectation in a first step: 
    \begin{equation*}
        	E\big[\sum_j w_j C_j \big] \stackrel{\eqref{eq: C_j from single to parallel machines}}{\leq} \frac{1}{m} E\big[\sum_j w_j \sum_{k\not= j} c(k,j) \big]+ \sum_j w_j E\big[p_j^A \big]
    \end{equation*}
    In a second step, one uses the upper bound for \textsc{Weighted} $(f,2)$-PCP on a single machine from Theorem~\ref{theorem: randomized bound single machine}, which is $\nicefrac{3(7+3\sqrt{6})}{20}\cdot \sum_j w_j \sum_{k\not= j} c^*(k,j)$ for the first expectation. Further, one calculates the job-specific bound $E\left[ p_j^A\right]\leq  p^*_j \cdot \nicefrac{(7+3\sqrt{6})}{10}$ for each job $j$ and for~$f$ as in~\eqref{eq: prob function} similar to the proof of Theorem~\ref{theorem: randomized bound single machine}. To this end, consider the following case distinction on the ranges of $r_j$ to give an upper bound on $E\left[p_j^A\right] = u_j \left( 1- f\left(r_j\right)\right) +\left(t_j+p_j\right)f\left(r_j\right)$. For $r_j<1$, we obtain $E\left[ p_j^A\right]=u_j=p_j^*$, and for $r_j>3$, $E\left[ p_j^A\right]=t_j + p_j = \nicefrac{u_j}{r_j}+p_j \leq \nicefrac{4}{3}\cdot u_j$, thus $E\left[ p_j^A\right] \leq \nicefrac{4}{3}\cdot p_j^*$. For $1\leq r_j\leq 3$, we obtain 
    \begin{equation*}
            \frac{E\left[ p_j^A\right]}{u_j} \leq f(r_j)\Big(\frac{1}{r_j}+1\Big) +\big(1-f(r_j)\big)\leq 1+\frac{f(r_j)}{r_j}\leq \frac{(7+3\sqrt{6})}{10} ,
    \end{equation*}
    since $1+\nicefrac{f(r_j)}{r_j}$ attains its maximum at $r_j=1+\sqrt{\nicefrac{2}{3}}$ with the above value, on the one hand and 
    \begin{equation*}
        \frac{E\left[ p_j^A\right]}{t_j+p_j} \leq f(r_j) + r_j\left(1-f(r_j)\right) \leq \frac{6}{5} ,
    \end{equation*}
    since the expression is maximized at $r_j=\nicefrac{3}{2}$ with value $\nicefrac{6}{5}$, on the other hand. In total, we have the bound $\nicefrac{E\left[ p_j^A\right]}{p_j^*} \leq \nicefrac{(7+3\sqrt{6})}{10}\approx 1.4348$. Substituting the bounds into the arguments as in the proof of Theorem~\ref{theorem: parallel machines proof}, we finally obtain $E\left[\sum_j w_j C_j\right]\leq \nicefrac{3(7+3\sqrt{6})}{20}\cdot OPT_m$, if $\nicefrac{(7+3\sqrt{6})}{10}\leq\nicefrac{3(7+3\sqrt{6})}{20}\cdot(\nicefrac{(m+1)}{2m})$, i.e. $m\leq 3$, and for each $m\geq 4$:
    \begin{equation*}
    \begin{aligned}
        E\Big[\sum_j w_j C_j\Big] &\leq OPT_m \frac{\left( 7 + 3\sqrt{6}\right)}{10}\Big( \frac{3}{2}\Big(\frac{1}{2}-\frac{1}{2m}\Big) + 1\Big) \\
        &\leq OPT_m \frac{\left( 7 + 3\sqrt{6}\right)}{40}\Big( 7 -\frac{3}{m}\Big)\leq OPT_m \cdot\Big(2.5110-\frac{1.0761}{m}\Big) \ .
    \end{aligned}
    \end{equation*}

    The conference version proved a competitive ratio of at most $\left( 7 - \nicefrac{1}{m}\right) \nicefrac{\left( 7 + 3\sqrt{6}\right)}{40}\leq 2.5110-\nicefrac{0.3587}{m}$ for each $m\geq 1$. In this version, we also give an improved bound for each fixed $m\geq 2$.
\end{proof}

\section{Conclusion}
In this paper, we give the first algorithms with constant competitive ratios for scheduling with testing on a single machine to minimize the total \emph{weighted} completion time objective. These algorithms are natural extensions of algorithms that had earlier been proposed for the unweighted case \cite{liu2023power}. We present upper bounds of $2.3166$ and $2.1523$ on a single machine, and of $2.7763$ and $2.5110$ on identical parallel machines, for the competitive ratio of a deterministic and a randomized algorithm, respectively.

An interesting open question is whether there is a gap between the problem without job-dependent weights and the one with them. While our results show that the currently best-known upper bounds established for the unweighted setting can be preserved, no separation between those cases is known at the moment. Hence, this remains a natural direction for future work, alongside tightening the gaps between the lower and upper bounds for this model in general.

\section*{Acknowledgments} Felix Buld was funded by the Deutsche Forschungsgemeinschaft (DFG, German Research Foundation) - GRK 2201/2 - Projektnummer 277991500. The authors thank Simon Gmeiner for helpful discussions and ideas on the $WGR^3$ rule from Section~\ref{section: From Preemptive to Test Preemptive Scheduling}. In addition, the authors thank the anonymous reviewers of IJTCS-FAW 2025 for their valuable feedback and suggestions.

\bibliographystyle{abbrvurl}
\bibliography{literature_lv}

\appendix

\section{Unboundedness of Previous Approaches Considering Job-Dependent Weights}
\label{subsec: unboundedness zhang}
For scheduling with testing to minimize the total \emph{weighted} completion time in the adversarial model, some algorithms were proposed by \cite{zhang2023scheduling}. They considered the restricted problem with uniform testing times $t_j=1$ and uniform upper bounds $u_j=u$, where jobs differ only by their weights. Three algorithms to approach the job-dependent weights were outlined:
\begin{enumerate}
    \item \emph{Greedy}: Sort jobs according to non-increasing weights $w_1 \geq w_2 \geq \ldots \geq w_n$. In this order, test a job and execute it immediately.
    \item \emph{Delay-All}: Test all jobs. Afterwards, execute jobs in non-decreasing order of their weighted processing times~$\nicefrac{p_1}{w_1}\leq \nicefrac{p_2}{w_2}\leq \ldots \leq \nicefrac{p_n}{w_n}$.
    \item Intermediate strategy \emph{L-Delay-All}: Group jobs in batches of same weight. Then apply \emph{Delay-All} individually to each batch in order of non-increasing weights.
\end{enumerate}
They provide upper bounds on the competitive ratios of these algorithms, parameterized by input values such as the upper bound~$u$ or the maximal weight~$w_{\max}$. It also contains lower bounds of $2.4$ for the second and $3.4$ for the third strategy. However, all those strategies have unbounded competitive ratios, even for uniform testing times and uniform upper bounds. We construct instance families that are structurally similar to theirs to show this. 
    \begin{enumerate}
    \item Algorithms \emph{Greedy} and \emph{L-Delay-All}:
        Consider instances with $n\geq 2$ \emph{small} jobs, one \emph{large} job and $\varepsilon>0$ small:
        \begin{equation*}
            \begin{aligned}
                &\textup{(\emph{short}) } &&t_j=1, u_j= n^2, p_j=0, \hspace{1ex} w_j=1 &&\textup{ for } j=1,\ldots,n,\\
                 &\textup{(\emph{long}) } &&t_j=1,u_j= n^2, p_j=n^2, w_j=1+\varepsilon && \textup{ for } j=n+1\ .
            \end{aligned}
        \end{equation*}
        
        Since the first algorithm focuses only on the weight, it processes the long job immediately after testing, causing a large delay for all short jobs. In contrast, the optimal offline algorithm does not test the job and schedules it last. Hence, it follows
        \begin{equation*}
            \frac{ALG}{OPT}=\frac{(1+\varepsilon)(1+n^2) + (1+n^2)n + \frac{n(n+1)}{2}}{\frac{n(n+1)}{2} + (1+\varepsilon)(n+n^2)} ,
        \end{equation*}
        which is unbounded for $n\rightarrow \infty$. The objective value for these instances, and hence the competitive ratio, is even worse for L-Delay-All, since the small jobs are delayed after their tests as well.
        
    \item  Algorithm \emph{Delay-All}: Consider instances with~$n$ \emph{light} jobs and one \emph{heavy} job:
        \begin{equation*}
            \begin{aligned}
                &\textup{(\emph{light}) } &&t_j=1, u_j= 2, p_j=0, w_j=1 &&\textup{ for } j=1,\ldots,n,\\
                 &\textup{(\emph{heavy}) } &&t_j=1,u_j= 2, p_j=0, w_j=n^2 && \textup{ for } j=n+1\ . 
            \end{aligned}
        \end{equation*}
        \emph{Delay-All} tests all jobs first and then executes them. Especially the heavy job is not completed immediately after testing. The optimal offline algorithm tests and executes each job immediately after the test, starting with the heavy one. Hence, it follows:
        \begin{equation*}
            \frac{ALG}{OPT}=\frac{(n+1)(n^2+n)}{n^2+n+\frac{n(n+1)}{2}} ,
        \end{equation*}
        which is unbounded for $n\rightarrow \infty$. Note that the \emph{Delay-All} algorithm may also be viewed as \textsc{Weighted} $(0,0)$-\textsc{SORT}, see Table~\ref{table: beta results}.
    \end{enumerate}

The instances constructed in this section show that the approaches of \cite{zhang2023scheduling} that focus on only one aspect, such as ordering by weights or early testing for fast information gain, are not sufficient to achieve constant competitive ratios. This even holds in the highly restrictive setting with uniform testing times and uniform upper bounds. In contrast, our approaches integrate testing decisions based on thresholds and priority rules that combine testing and processing times with job-dependent weights, and achieve constant competitive ratios in the general setting. 

\section{Upper Bound of the \textsc{Weighted} $(\alpha,\beta)$-\textsc{SORT}}
\label{seq: upper bound SORT}
The analysis of the $(\alpha,\beta)$-\textsc{SORT} in \cite{liu2023power} can be extended to the weighted version analogously as for the \textsc{Weighted} $(\alpha,\beta)$-\textsc{PCP} in Lemma~\ref{lemma: test amortization} and Theorem~\ref{theorem competitive ratio beta}  giving rise to the \textsc{Weighted} $(\alpha,\beta)$-\textsc{SORT}. We illustrate this in this section.

\begin{lemma} 
\label{lemma: SORT test amortization}
Consider jobs $j\not = k$, and parameters $\alpha\geq 1$ and $\beta >0$. 
\\ If the \textsc{Weighted} $(\alpha,\beta)$-\textsc{SORT} does not test job $k$, it holds
\begin{equation*}
     w_j \cdot c(k,j)+w_k \cdot c(j,k) \leq w_j \cdot \Big(1+ \frac{1}{\beta}\Big) u_k \leq w_j \cdot \alpha\cdot \Big(1+ \frac{1}{\beta}\Big) \cdot p^*_k \ .
\end{equation*}
If the \textsc{Weighted} $(\alpha,\beta)$-\textsc{SORT} does test job $k$, it holds
\begin{equation*}
    \begin{aligned}
        w_j \cdot c(k,j)+w_k \cdot c(j,k) &\leq w_j \cdot \max\bigg\lbrace \Big( 1+\beta\Big) t_k, 2t_k + p_k, t_k+\Big(1+\frac{1}{\beta}\Big)p_k\bigg\rbrace \\&\leq w_j \cdot \max\bigg\lbrace 1+\beta, 2, 1+\frac{2}{\alpha}, 1+\frac{1}{\alpha} + \frac{1}{\beta}\bigg\rbrace \cdot p^*_k \ .
    \end{aligned}
\end{equation*}
\end{lemma}
\begin{proof}
    The proof naturally extends the proof of~\cite{liu2023power} to job-dependent weights and follows the same structure as the proof of Lemma~\ref{lemma: test amortization} using an exhaustive comparison tree as in \cite{albers2021explorable,dogeas2024obligatory}.
    
    Assume that $k<_o j$ for a given pair of jobs $k,j$. We bound the weighted sum of mutual contributions $w_j \cdot c(k,j)+w_k \cdot c(j,k)$ in the algorithm against the value in the optimal offline solution, i.e., $w_j \cdot p^*_k$.
    We use the notation from \cite{dogeas2024obligatory} as extended in Lemma~\ref{lemma: test amortization} and denote the testing task by $\tau_j$, the tested-execution task by $\pi_j$, and the untested-execution task of job~$j$ by $\upsilon_j$. An ordered list, such as~$\left[ \upsilon_k, \tau_j,\pi_j\right]$ describes the testing and ordering decisions of the algorithm for jobs $k,j$ compactly. 
    
    First, we consider the cases in which the algorithm does not test job~$k$, see Table~\ref{table:SORT no test on k}. Overall, it follows that $w_j\cdot c(k,j)+ w_k \cdot c(j,k) \leq w_j \cdot (1+ \nicefrac{1}{\beta}) u_k$ in this case. Bounding $u_k\leq\alpha p_k^*$ by Lemma~\ref{lemma: single job testing} yields the stated bound.

{\small 
        \begin{table}[htbp]
                \centering
                \caption{Case $p_k^A=u_k$\label{table:SORT no test on k}}
                \vspace{2ex}
                \begin{tabular}{c|c|l|c}
                     &Order of Tasks & $w_j\cdot c(k,j)+w_k \cdot c(j,k)$ & Used Inequalities \\
                     \hline
                     \multirow{2}{*}{$p_j^A=u_j$}
                     & $\left[ \upsilon_k,\upsilon_j\right]$ & $ w_j\cdot u_k$ & \\
                     
                     &$\left[\upsilon_j,\upsilon_k\right]$  & $ w_k\cdot u_j \leq w_j\cdot u_k$ & $\frac{u_j}{w_j}\leq \frac{u_k}{w_k}$\\
                     \hline
                     \multirow{3}{*}{$p_j^A=t_j + p_j$} 
                     & $\left[ \upsilon_k,\tau_j, \pi_j\right]$ & $ w_j\cdot u_k$& \\
                     
                     & $\left[ \tau_j,\upsilon_k, \pi_j\right]$ & $ w_j\cdot u_k + w_k\cdot t_j \leq \big( 1+ \frac{1}{\beta} \big) w_j\cdot u_k$& $\frac{\beta t_j}{w_j}\leq\frac{u_k}{w_k}$\\
                     
                     & $\left[ \tau_j, \pi_j ,\upsilon_k\right]$ & $ w_k\cdot (t_j + p_j) \le \big( 1+ \frac{1}{\beta} \big) w_j\cdot u_k $ & $\frac{\beta t_j}{w_j}\leq\frac{u_k}{w_k}$; $\frac{p_j}{w_j}\leq\frac{u_k}{w_k}$
                \end{tabular}
            \end{table}
            }

  Secondly, we consider the cases when the algorithm tests job $k$, see Table~\ref{table:SORT test on k}. We obtain $w_j  c(k,j) + w_k  c(j,k) \leq w_j \cdot \max\left\lbrace \left( 1+\beta\right) t_k, 2t_k + p_k, t_k+\left(1+\nicefrac{1}{\beta}\right)p_k\right\rbrace$ in these cases, as stated above. We further bound the expressions with respect to the optimal running time $p^*_k$. First, we have~$\left(1+\beta\right) t_k \leq \left(1+\beta\right) p^*_k$, since~$p^*_k=\min\lbrace t_k+p_k, u_k\rbrace$ and $t_k\leq \alpha t_k \leq u_k$. Secondly, if $p^*_k=t_k+p_k$, we have $2t_k + p_k\leq 2t_k + 2p_k = 2p^*_k$. If, however, $p^*_k=u_k$, we have $2t_k + p_k = t_k + t_k+ p_k \leq \nicefrac{p^*_k}{\alpha} + \left(1+\nicefrac{1}{\alpha}\right)
  p^*_k= \left(1+\nicefrac{2}{\alpha}\right)p^*_k$ by Lemma~\ref{lemma: single job testing} and $t_k\leq \nicefrac{u_k}{\alpha}$ since $k$ was tested.
  Lastly, if $p^*_k=t_k+p_k$, we have $t_k+\left(1+\nicefrac{1}{\beta}\right)
    p_k\leq \left(1+\nicefrac{1}{\beta}\right)p^*_k$. And, if $p^*_k=u_k$, we have $t_k+\left(1+\nicefrac{1}{\beta}\right)
    p_k = t_k + p_k + \nicefrac{1}{\beta}p_k \leq \left(1+\nicefrac{1}{\alpha}+\nicefrac{1}{\beta}\right)p^*_k$ by Lemma~\ref{lemma: single job testing} and $p_k\leq u_k$. These calculations prove the second inequality in the statement of the Lemma.
{\small
    \begin{table}[htbp]
        \centering
        \caption{Case $p_k^A= t_k + p_k$\label{table:SORT test on k}}
        \vspace{2ex}
        \begin{tabular}{c|c|l|c}
             &Order of Tasks & $w_j\cdot c(k,j)+w_k \cdot c(j,k)$ & Used Inequalities \\
             \hline
            \multirow{3}{*}{\rotatebox[origin=c]{90}{$p_j^A=u_j$}}
            &$\left[\tau_k, \pi_k, \upsilon_j \right]$ & $w_j\cdot (t_k+p_k)$ & \\
            &$\left[ \tau_k,\upsilon_j,\pi_k\right]$ & $ w_j\cdot t_k + w_k\cdot u_j \leq w_j\cdot (t_k+ p_k)$ & $\frac{u_j}{w_j}\leq\frac{p_k}{w_k}$\\
            &$\left[ \upsilon_j, \tau_k, \pi_k \right]$ & $ w_k\cdot u_j \leq \beta\cdot w_j \cdot t_k$ & $\frac{u_j}{w_j}\leq \frac{\beta t_k}{w_k}$\\
            \hline
            \multirow{6}{*}{\rotatebox[origin=c]{90}{$p_j^A=t_j + p_j$}}
                 &$\left[\tau_k, \pi_k, \tau_j, \pi_j \right]$ & $ w_j\cdot (t_k+p_k)$& \\
                 
                 &$\left[\tau_k,  \tau_j, \pi_k, \pi_j \right]$ & $ w_j\cdot (t_k+p_k) + w_k\cdot t_j \leq w_j\cdot \big(t_k+\big(1+\frac{1}{\beta}\big)p_k\big)$ & $\frac{\beta t_j}{w_j}\leq\frac{p_k}{w_k}$\\ 
                 
                 &$\left[\tau_k,  \tau_j, \pi_j, \pi_k \right]$ &$ w_j\cdot t_k + w_k\cdot (t_j+p_j) \leq w_j\cdot \big(t_k+\big(1+\frac{1}{\beta}\big)p_k\big)$ & $\frac{\beta t_j}{w_j}\leq\frac{p_k}{w_k};\frac{p_j}{w_j}\leq\frac{p_k}{w_k}$\\
                 
                 &$\left[\tau_j, \tau_k, \pi_k, \pi_j \right]$ & $ w_j\cdot (t_k+p_k) + w_k\cdot t_j \leq w_j\cdot \left(2t_k+p_k\right)$ & $\frac{\beta t_j}{w_j}\leq\frac{\beta t_k}{w_k}$ \\
                 
                 &$\left[\tau_j, \tau_k, \pi_j, \pi_k \right]$ & $ w_j\cdot t_k + w_k\cdot (t_j+p_j) \leq w_j\cdot (2t_k+ p_k)$& $\frac{\beta t_j}{w_j}\leq \frac{\beta t_k}{w_k}; \frac{p_j}{w_j}\leq \frac{p_k}{w_k}$\\
                 
                 &$\left[\tau_j, \pi_j , \tau_k, \pi_k \right]$ & $ w_k\cdot(t_j+p_j) \leq \left(1+\beta\right)\cdot w_j \cdot t_k$& $\frac{\beta t_j}{w_j}\leq \frac{\beta t_k}{w_k}; \frac{p_j}{w_j}\leq \frac{\beta t_k}{w_k}$
            \end{tabular}
        \end{table}
        }
\end{proof}

Note that in Lemma~\ref{lemma: SORT test amortization} some appearing terms in the case where a job $k$ is tested differ in comparison to Lemma~\ref{lemma: test amortization} for the \textsc{Weighted} $(\alpha,\beta)$-\textsc{PCP} because of the different priority values. Thus, the following theorem provides other optimal parameters for this expression as for the corresponding one in Theorem~\ref{theorem competitive ratio beta}. 

\begin{theorem}
\label{theorem SORT competitive ratio beta}
    The competitive ratio of the \textsc{Weighted} $(\alpha,\beta)$-\textsc{SORT} with $\alpha\geq 1$ and $\beta>0$ is at most 
    \begin{equation*}
    Q_2(\alpha,\beta) :=  \max\bigg\lbrace \alpha\Big(1+\frac{1}{\beta}\Big), 1+\beta, 2, 1+\frac{2}{\alpha}, 1+\frac{1}{\alpha} + \frac{1}{\beta}\bigg\rbrace \ .
    \end{equation*}
    The optimal choice of the parameters to minimize this term is $\alpha=\beta=\sqrt{2}$, resulting in a competitive ratio of at most $1+\beta \leq 2.4143$.

\end{theorem}
\begin{proof}
    Analogously to Theorem~\ref{theorem competitive ratio beta} we get
    \begin{equation*}
    \label{eq: competitive ratio SORT}
    \begin{aligned}
        \sum_{j=1}^{n} w_j\cdot C_j^A &= \sum_{j=1}^{n} \Big( w_j \cdot c(j,j) +  \sum_{k<_o j} \big( w_j \cdot c(k,j)+w_k \cdot c(j,k)\big)\Big)
        \\&\leq \sum_{j=1}^{n} \Big(w_j\cdot p_j^A + \sum_{k<_o j} \big(w_j \cdot Q_2(\alpha,\beta)\cdot p^*_k\big)\Big)
        \\&\leq Q_2(\alpha,\beta)\cdot \sum_{j=1}^{n} \sum_{k\leq_o j} w_j\cdot  p^*_k = Q_2(\alpha,\beta)\cdot OPT \ .
    \end{aligned}
    \end{equation*} 
    The first inequality follows from Lemma~\ref{lemma: SORT test amortization}, the second one holds, since $p_j^A\leq\max\lbrace \alpha, 1+\nicefrac{1}{\alpha}\rbrace\cdot p^*_j \leq Q_2(\alpha,\beta) \cdot p^*_j$. Minimizing this upper bound over $\alpha$ and $\beta$ yields $\alpha=\beta =\sqrt{2}$ as optimal values, cf.~\cite{liu2023power}. 
    For the choice of parameters $\alpha=\beta=\sqrt{2}$, all parametrized terms in the definition of $Q_{2}(\alpha,\beta)$ are equal to $1+\sqrt{2}$. To improve upon this, one would need $\beta<\sqrt{2}$ since we have the term $1+\beta$ in the maximum, hence $\alpha<\sqrt{2}$ since we have the term $\alpha(1+\nicefrac{1}{\beta})$, but then $1+\nicefrac{2}{\alpha}> 1+\sqrt{2}$.
\end{proof}

\begin{remark}
    The \textsc{Weighted} $(\alpha,\beta)$-\textsc{SORT} algorithm can be extended to identical parallel machines as the \textsc{Weighted} $(\alpha,\beta)$-\textsc{PCP} algorithm in Section~\ref{section: Parallel Machines}. Using the pairwise contribution bound from Theorem~\ref{theorem SORT competitive ratio beta} and applying the same framework as in Theorem~\ref{theorem: parallel machines proof}, one can not achieve a better bound than $\nicefrac{Q_2(\sqrt{2},\sqrt{2})}{2}+\max\lbrace\varphi,1+\nicefrac{1}{\varphi}\rbrace=\nicefrac{(1+\sqrt{2})}{2}+\varphi\geq 2.8251$.
\end{remark}

\end{document}